\documentclass[a4paper,UKenglish]{lipics-v2016}
\usepackage{hyperref}

\usepackage[T1]{fontenc}

\usepackage[utf8]{inputenc}

\usepackage{microtype}

\usepackage{pgfplots}

\usepackage{xspace,relsize,stmaryrd}
\usepackage{amssymb,amsmath}
\usepackage{mathpartir}
\usepackage{bussproofs}
\usepackage{caption}
\usepackage{subcaption}

\allowdisplaybreaks[2]

\usepackage{tikz}
 \usetikzlibrary{shapes}
 \usetikzlibrary{fit}
 \usetikzlibrary{shadows}
 \usetikzlibrary{backgrounds}
 \usetikzlibrary{arrows}
\newcounter{blubber}

\newcommand{\rank}{\mathsf{rk}}

\newcommand{\concl}{\mathit{Cn}}

\newcommand{\Rec}{\mathit{Inh}}
\newcommand{\defer}{\mathit{dfr}}
\newcommand{\PLentails}{\vdash_{\mathit{PL}}}
\newcommand{\FLphi}{\mathbf F}
\newcommand{\snodes}{\mathbf{S}}
\newcommand{\nodes}{\mathbf{N}}
\newcommand{\fnodes}{\mathbf{C}}

\newcommand\NExpTime{$\textsc{NExpTime}$\xspace}
\newcommand\ExpTime{$\textsc{ExpTime}$\xspace}

\newcommand{\Ni}{\noindent}
\newcommand{\Sem}[1]{{[\![#1]\!]}}

\newcommand{\sem}[1]{[\![#1]\!]}

\newcommand{\psem}[1]{\widehat{[\![#1]\!]}}

\newcommand{\Pow}{\mathcal{P}}

\newcommand{\lO}{\mathcal{O}}
\newcommand{\infrule}[2]{\frac{#1}{#2}}

\newcommand{\sut}{such that\xspace}

\newenvironment{alg}[1][Algorithm]{\begin{trivlist}
\item[\hskip \labelsep {\bfseries Algorithm (#1).}]}{\end{trivlist}}

\begin{document}

\bibliographystyle{plainurl}

\title{Global Caching for the Alternation-free $\mu$-Calculus}

\author{Daniel Hausmann, Lutz Schr\"oder, and Christoph Egger}
\affil{Friedrich-Alexander Universit\"{a}t
Erlangen-N\"urnberg, Germany}


\maketitle

\begin{abstract}
  We present a sound, complete, and optimal single-pass tableau
  algorithm for the alternation-free $\mu$-calculus. The algorithm
  supports global caching with intermediate propagation and runs in
  time $2^{\lO(n)}$. In game-theoretic terms, our algorithm integrates
  the steps for constructing and solving the B\"uchi game arising from
  the input tableau into a single procedure; this is done
  on-the-fly, i.e.\ may terminate before the game has been fully
  constructed. This suggests a slogan to the effect that \emph{global
    caching = game solving on-the-fly}. A prototypical implementation
  shows promising initial results.
\end{abstract}

\section{Introduction}

\noindent The modal $\mu$-calculus~\cite{Kozen83,BradfieldStirling06}
serves as an expressive temporal logic for the specification of
sequential and concurrent systems containing many standard formalisms
such as linear time temporal logic LTL~\cite{MannaPnueli79,Pnueli77},
CTL~\cite{ClarkeEmerson81}, and
PDL~\cite{Pratt76}. 
Satisfiability checking in the modal $\mu$-calculus is
$\ExpTime$-complete~\cite{NiwinskiWalukiewicz96,EmersonJutla99}. There
appears to be, to date, no readily implementable reasoning algorithm
for the $\mu$-calculus, and in fact (prior
to~\cite{HausmannSchroeder15}) even for its fragment CTL, that is
simultaneously \emph{optimal}, i.e.\ runs in \ExpTime, and
\emph{single-pass}, i.e.\ avoids building an exponential-sized data
structure in a first pass. Typical data structures used in
worst-case-optimal algorithms are automata~\cite{EmersonJutla99},
games~\cite{FriedmannLange13a}, and, for sublogics such as CTL,
first-pass tableaux~\cite{EmersonHalpern85}.

The term \emph{global caching} describes a family of single-pass
tableau algorithms~\cite{
  GoreNguyen13,GoreWidmann09} that build graph-shaped tableaux
bottom-up in so-called \emph{expansion} steps, with no label ever
generated twice, and attempt to terminate before the tableau is
completely expanded by means of judicious intermediate
\emph{propagation} of satisfiability and/or unsatisfiability through
partially expanded tableaux. Global caching offers wide room for
heuristic optimization, regarding standard tableau optimizations as
well as the order in which expansion and propagation steps are
triggered, and has been shown to perform competitively in practice;
see~\cite{GoreWidmann09} for an evaluation of heuristics in global
caching for the description logic $\mathcal{ALCI}$. One major
challenge with global caching algorithms is typically to prove
soundness and completeness, which becomes harder in the presence of
fixpoint operators. A global caching algorithm for PDL has been
described by Gor\'e and Widmann~\cite{GoreWidmann09pdl}; finding an
optimal global caching algorithm even for CTL has been named as an
open problem as late as 2014~\cite{Gore14} (a non-optimal, doubly
exponential algorithm is known~\cite{Gore14}).

The contribution of the present work is an optimal global-caching
algorithm for satisfiability in the alternation-free $\mu$-calculus, extending our
earlier work on the single-variable (\emph{flat}) fragment of the
$\mu$-calculus~\cite{HausmannSchroeder15}.  The algorithm actually
works at the level of generality of the alternation-free fragment of
the coalgebraic $\mu$-calculus~\cite{CirsteaEA11}, and thus covers
also logics beyond the realm of standard Kripke semantics such as
alternating-time temporal logic ATL~\cite{AlurEA02}, neighbourhood-based
logics such as the monotone $\mu$-calculus that underlies Parikh's
game logic~\cite{Parikh85}, or probabilistic fixpoint logic. To aid
readability, we phrase our results in terms of the relational
$\mu$-calculus, and discuss the coalgebraic generalization only at the
end of Section~\ref{section:soundcomp}. The model construction in the
completeness proof yields models of size $2^{\lO(n)}$.

We have implemented of our algorithm as an extension of the
Coalgebraic Ontology Logic Reasoner COOL, a generic reasoner for
coalgebraic modal logics~\cite{GorinEA14}; given the current state of
the implementation of instance logics in COOL, this means that we
effectively support alternation-free fragments of relational,
monotone, and alternating-time~\cite{AlurEA02} $\mu$-calculi, thus in
particular covering CTL and ATL. We have evaluated the tool in
comparison with existing reasoners on benchmark formulas for
CTL~\cite{GoreEA11} (which appears to be the only candidate logic for
which well-developed benchmarks are currently available) and
on random formulas for ATL and the alternation-free relational $\mu$-calculus, 
with promising results; details are discussed in
Section~\ref{section:cool}.

\subparagraph*{Related Work} The theoretical upper bound $\ExpTime$
has been established for the full coalgebraic
$\mu$-calculus~\cite{CirsteaEA11} (and earlier for instances such as
the alternating-time $\mu$-calculus AMC~\cite{ScheweThesis}), using a
multi-pass algorithm that combines games and automata in a similar way
as for the standard relational case, in particular involving the Safra
construction. \emph{Global caching} has been employed successfully for
a variety of description logics~\cite{
  GoreNguyen13,GoreWidmann09
}, and lifted to the level of generality of coalgebraic logics with
global assumptions~\cite{GoreEA10a} and nominals~\cite{GoreEA10b}.

A tableaux-based non-optimal (\NExpTime) decision procedure for the
full $\mu$-calculus has been proposed in~\cite{Jungteerapanich09}. 
Friedmann and Lange~\cite{FriedmannLange13a} describe an optimal
tableau method for the full $\mu$-calculus that, unlike most other
methods including the one we present here, makes do without requiring
guardedness. Like earlier algorithms for the full $\mu$-calculus, the
algorithm constructs and solves a parity game, and in principle allows
for an on-the-fly implementation. The models constructed in the
completeness proof are asymptotically larger than ours, but presumably
the proof can be adapted for the alternation-free case by using
determinization of co-B\"uchi automata~\cite{MiyanoHayashi1984}
instead of Safra's determinization of B\"uchi automata~\cite{Safra88}
to yield models of size $2^{\lO(n)}$, like ours. For non-relational
instances of the coalgebraic $\mu$-calculus, including the
alternation-free fragment of the alternating-time $\mu$-calculus AMC,
the $2^{\lO(n)}$ bound on model size appears to be new, with the best
known bound for the alternation-free AMC being
$2^{\lO(n\log n)}$~\cite{ScheweThesis}.

In comparison to our own recent work~\cite{HausmannSchroeder15}, we
move from the flat to the alternation-free fragment, which means
essentially that fixpoints may now be defined by mutual recursion, and
thus can express properties such as `all paths reach states satisfying
$p$ and $q$, respectively, in strict alternation until they eventually
reach a state satisfying~$r$'.  Technically, the main additional
challenge is the more involved structure of eventualities and
deferrals, which now need to be represented using cascaded sequences
of unfoldings in the focusing approach; this affects mainly the
soundness proof, which now needs to organize termination counters in a
tree structure.  While the alternation-free algorithm instantiates to
the algorithm from~\cite{HausmannSchroeder15} for flat input formulas,
its completeness proof includes a new model construction which
yields a bound of $3^n\in 2^{\lO(n)}$ on model size, slightly
improving upon the bound $n\cdot 4^n$
from~\cite{HausmannSchroeder15}. We present the new algorithm in terms
that are amenable to a game-theoretic perspective, emphasizing the
correspondence between global gaching and game-solving. In fact, it
turns out that global caching algorithms effectively consist in an
integration of the separate steps of typical game-based methods for
the
$\mu$-calculus~\cite{FriedmannLange13a,FriedmannLatteLange13,NiwinskiWalukiewicz96}
into a single on-the-fly procedure that talks only about partially
expanded tableau graphs, implicitly combining on-the-fly
determinization of co-B\"uchi automata with on-the-fly solving of the
resulting B\"uchi games~\cite{FriedmannLange10}. This motivates the
mentioned slogan that\vspace{-0.5ex}
\begin{quote}
  \emph{global caching is on-the-fly determinization and game
    solving}.
\end{quote}
In particular, the propagation steps in the global caching pattern can
be seen as solving an incomplete Büchi game that is built directly by
the expansion steps, avoiding explicit determinization of co-B\"uchi
automata analogously to~\cite{MiyanoHayashi1984}. One benefit of an
explicit global caching algorithm integrating the pipeline from
tableaux to game solving is the implementation freedom afforded by the
global caching pattern, in which suitable heuristics can be used to
trigger expansion and propagation steps in any order that looks
promising.


\section{Preliminaries: The $\mu$-Calculus}\label{section:prelim}

\Ni We briefly recall the definition of the (relational)
$\mu$-calculus. We fix a set $P$ of \emph{propositions}, a set $A$ of
\emph{actions}, and a set $\mathfrak{V}$ of fixpoint
variables. Formulas $\phi,\psi$ of the $\mu$-calculus are then defined
by the grammar
\begin{align*}
\psi,\phi ::= \bot \mid \top \mid p \mid \neg p \mid X \mid \psi\wedge\phi \mid \psi\vee\phi
\mid \langle a \rangle\psi \mid [a]\psi \mid \mu X.\, \psi \mid \nu X.\, \psi
\end{align*}
where $p\in P$, $a\in A$, and $X\in \mathfrak{V}$; we write $|\psi|$
for the size of a formula $\psi$. Throughout the paper, we use $\eta$
to denote one of the fixpoint operators $\mu$ or $\nu$. We refer to
formulas of the form $\eta X.\,\psi$ as \emph{fixpoint literals}, to
formulas of the form $\langle a\rangle\psi$ or $[a]\psi$ as
\emph{modal literals}, and to $p$, $\neg p$ as \emph{propositional
  literals}. The operators $\mu$ and $\nu$ \emph{bind} their
variables, inducing a standard notion of \emph{free variables} in
formulas. We denote the set of free variables of a formula $\psi$ by
$FV(\psi)$. A formula $\psi$ is \emph{closed} if $FV(\psi)=\emptyset$,
and \emph{open} otherwise. We write $\psi\leq\phi$ ($\psi<\phi$) to
indicate that $\psi$ is a (proper) subformula of $\phi$. We say that
$\phi$ \emph{occurs free} in $\psi$ if $\phi$ occurs as a subformula
in $\psi$ that is not in the scope of any fixpoint.  Throughout, we
\emph{restrict to formulas that are guarded}, i.e.\ have at least one
modal operator between any occurrence of a variable $X$ and an
enclosing binder $\eta X$.  (This is standard although possibly not
without loss of generality~\cite{FriedmannLange13a}.)  Moreover we
assume w.l.o.g.\ that input formulas are \emph{clean}, i.e.\ all
fixpoint variables are distinct, and \emph{irredundant}, i.e.\
$X\in FV(\psi)$ for all subformulas $\eta X.\,\psi$.

Formulas are evaluated over \emph{Kripke structures}
$\mathcal{K}=(W,(R_a)_{a\in A},\pi)$, consisting of a set $W$ of
\emph{states}, a family $(R_a)_{a\in A}$ of relations
$R_a\subseteq W\times W$, and a valuation $\pi:P\to\Pow(W)$ of the
propositions. Given an \emph{interpretation}
$i:\mathfrak{V}\to \Pow(W)$ of the fixpoint variables, define
$\Sem{\psi}_i\subseteq W$ by the obvious clauses for Boolean operators
and propositions,
$\sem{X}_i=i(X)$, 
$\sem{\langle a\rangle\psi}_i=\{v\in W\mid \exists w\in R_a(v).w\in
\sem{\psi}_i\}$,
$\sem{[a]\psi}_i=\{v\in W\mid \forall w\in R_a(v).w\in\sem{\psi}_i\}$,
$\sem{\mu X.\,\psi}_i =\mu\sem{\psi}^X_i$ and
$\sem{\nu X.\,\psi}_i =\nu\sem{\psi}^X_i$, where
$R_a(v)=\{w\in W\mid (v,w)\in R_a\}$,
$\sem{\psi}^X_i(G) = \sem{\psi}_{i[X\mapsto G]}$, and $\mu$, $\nu$
take least and greatest fixpoints of monotone functions, respectively.
If $\psi$ is closed, then $\sem{\psi}_i$ does not depend on $i$, so we
just write $\sem{\psi}$.  We write $x\models\psi$ for
$x\in\sem{\psi}$. The \emph{alternation-free fragment} of the
$\mu$-calculus is obtained by prohibiting formulas in which some
subformula contains both a free $\nu$-variable and a free
$\mu$-variable. E.g.\
$\mu X.\,\mu Y.\,(\Box X\land\Diamond Y\land\nu Z.\,\Diamond Z)$ is
alternation-free but 
$\nu Z.\,\mu X.\,(\Box X\land \nu Y.\,(\Diamond Y\land\Diamond Z))$ is
not.  CTL is contained in the alternation-free fragment.

We have the standard \emph{tableau rules} (each consisting of one
\emph{premise} and a possibly empty set of \emph{conclusions}) which
will be interpreted AND-OR style, i.e.\ to show satisfiability of a set
of formulas $\Delta$, it will be necessary to show that \emph{every}
rule application that matches $\Delta$ has \emph{some} conclusion that
is satisfiable. Our algorithm will use these rules in the expansion
step.
		\begin{align*}
  (\bot)\quad & \;\;\;\;\quad\quad\infrule{\Gamma,\bot}
	{}
&  (\lightning)\quad & \quad\quad\infrule{\Gamma,p,\neg p}
	{}
\\
  (\wedge)\quad & \;\quad\quad\infrule{\Gamma,\psi\wedge \phi}
	{\Gamma,\psi,\phi}
  & (\vee) \quad & \quad\infrule{\Gamma,\psi\vee \phi}
	{\Gamma,\psi \qquad \Gamma,\phi}
	\\[5pt]
  (\langle a\rangle) \quad & \infrule{\Gamma,[a] \psi_1,\ldots,[a]\psi_n,\langle a\rangle \phi}
	{\psi_1,\ldots,\psi_n,\phi}
  &(\eta) \quad &\infrule{\Gamma,\eta X.\, \psi}
	{\Gamma,\psi [X\mapsto \eta X.\, \psi]}
\end{align*}
(for $a\in A$, $n\in\mathbb{N}$, $p\in P$); we refer to the set of
modal rules $(\langle a\rangle)$ by $\mathcal{R}_m$ and to the set of
the remaining rules by $\mathcal{R}_p$ and usually write rules with
premise $\Gamma$ and conclusion $\Sigma=\Gamma_1,\ldots,\Gamma_n$ in sequential form,
i.e.\ as $(\Gamma/\Sigma)$.
\begin{example}
  As our running example, we pick a non-flat formula, i.e.\ one that
  uses two recursion variables. Consider the 
 alternation-free formulas
\begin{align*}
\psi_1 &= \mu X.\, ((p\wedge (r\vee \Box \psi_2)) \vee (\neg q\wedge \Box X)) &
\psi_2 &= \mu Y.\, ((q\wedge (r \vee \Box X)) \vee (\neg p\wedge \Box Y))
\end{align*}
(where $A=\{*\}$ and we write $\Box=[*]$,
$\Diamond=\langle *\rangle$). The formulas $\psi_1$ and $\psi_2[X\mapsto\psi_1]$ state
that all paths will visit $p$ and $q$ in strict alternation until $r$
is eventually reached, starting with $p$ and with~$q$,
respectively. 
\label{example:altfree}
\end{example}

\section{The Global Caching Algorithm}\label{section:algorithm}

\noindent We proceed to describe our global caching algorithm for the
alternation-free $\mu$-calculus. First off, we need some syntactic
notions regarding decomposition of fixpoint literals.
\begin{definition}[Deferrals]\label{defn:affil}
  Given fixpoint literals $\chi_i = \eta X_i.\,\psi_i$, $i=1,\dots,n$,
  we say that a substitution
  $\sigma=[X_1\mapsto \chi_1];\ldots;[X_n\mapsto\chi_n]$ \emph{sequentially
    unfolds $\chi_n$} if $\chi_i <_f \chi_{i+1}$ for all $1\leq i<n$,
  where we write $\psi <_f \eta X.\,\phi$ if $\psi\leq\phi$ and $\psi$
  is open and occurs free in $\phi$ (i.e.\ $\sigma$ unfolds a nested
  sequence of fixpoints in $\chi_n$ innermost-first).  We say that a
  formula $\chi$ is \emph{irreducible} if for every substitution
  $[X_1\mapsto \chi_1];\ldots;[X_n\mapsto \chi_n]$ that sequentially unfolds
  $\chi_n$, we have that
  $\chi = \chi_1([X_2\mapsto \chi_2];\ldots;[X_n\mapsto \chi_n])$
  implies $n=1$ (i.e.\ $\chi=\chi_1$).  An \emph{eventuality} is an
  irreducible closed least fixpoint literal.  A formula $\psi$
  \emph{belongs} to an eventuality $\theta_n$, or is a
  \emph{$\theta_n$-deferral}, if $\psi=\alpha\sigma$ for some
  substitution
  $\sigma = [X_1\mapsto \theta_1];\ldots;[X_n\mapsto \theta_n]$ that
  sequentially unfolds $\theta_n$ and some $\alpha <_f \theta_1$.  We denote
  the set of $\theta_n$-deferrals by $\mathit{dfr}(\theta_n)$.
\end{definition}
\Ni E.g.\ the substitution
$\sigma=[Y\mapsto \mu Y.\,(\Box X\land\Diamond\Diamond
Y)];[X\mapsto\theta]$
sequentially unfolds the eventuality
$\theta=\mu X.\,\mu Y.\,(\Box X\land\Diamond\Diamond Y)$, and
$(\Diamond Y)\sigma=\Diamond\mu Y.\,(\Box\theta\land\Diamond\Diamond
Y)$
is a $\theta$-deferral. A fixpoint literal is irreducible if it is not
an unfolding $\psi[X\mapsto\eta X.\,\psi]$ of a fixpoint literal
$\eta X.\,\psi$; in particular, every clean irredundant fixpoint
literal is irreducible. 

\begin{lemma}
  Each formula $\psi$ belongs to at most one eventuality $\theta$, and
  then $\theta\leq\psi$.
\label{lemma:afdefcontains}
\end{lemma}

\begin{example}\label{example:proofs}
Applying the tableau rules $\mathcal{R}_m$ and $\mathcal{R}_p$ to the formula 
$\psi_1\wedge EG\,\neg r$, where $\psi_1$ is defined as in Example~\ref{example:altfree}
and $EG\,\phi$ abbreviates $\nu X.\,(\phi\wedge \Diamond X)$,
results in a cyclic graph, with relevant parts depicted as follows: 
\begin{small}
\begin{prooftree}
\alwaysRootAtTop
\AxiomC{$\Gamma_1$}
\LeftLabel{\scriptsize$(\Diamond)$}
\UnaryInfC{$\Gamma,q,\Box \psi_1$}
\AxiomC{$\Gamma_2$}
\RightLabel{\scriptsize$(\Diamond)$}
\UnaryInfC{$\Gamma,\neg p,\Box \psi_2[X\mapsto \psi_1]$}
\LeftLabel{\scriptsize$(\vee,\wedge,\nu,\mu)^*$}
\BinaryInfC{$\psi_2[X\mapsto \psi_1], EG\,\neg r=:\Gamma_2$}
\LeftLabel{\scriptsize$(\Diamond)$}
\UnaryInfC{$\Gamma,p,\Box \psi_2[X\mapsto \psi_1]$}
\AxiomC{$\Gamma_1$}
\RightLabel{\scriptsize$(\Diamond)$}
\UnaryInfC{$\Gamma,\neg q, \Box \psi_1$}
\LeftLabel{\scriptsize$(\vee,\wedge,\nu,\mu)^*$}
\BinaryInfC{$\psi_1, EG\,\neg r=:\Gamma_1$}
\LeftLabel{\scriptsize$(\wedge)$}
\UnaryInfC{$\psi_1\wedge EG\,\neg r$}
\end{prooftree}
\end{small}
where $\Gamma =\{\neg r,\Diamond EG\,\neg r\}$. The graph contains
three cycles, all of which contain but never \emph{finish} a formula
that belongs to $\psi_1$ (where a formula belonging to an eventuality
$\psi_1$ is said to be \emph{finished} if it evolves to a formula that
does not belong to $\psi_1$): In the rightmost cycle, the deferral
$\delta_1:=\psi_1$ evolves to the deferral $\delta_2:=\Box \psi_1$
which then evolves back to $\delta_1$. For the cycle in the middle,
$\delta_1$ evolves to $\delta_3:=\Box \psi_2[X\mapsto \psi_1]$ which
in turn evolves to $\delta_4 := \psi_2[X\mapsto \psi_1]$ before
looping back to $\delta_3$. In the leftmost cycle, $\delta_1$ evolves
via $\delta_3$ and $\delta_4$ to $\delta_2$ before cycling back to
$\delta_1$. The satisfaction of $\psi_1$ is thus being postponed
indefinitely, since $EG\,\neg r$ enforces the existence of a path on
which $r$ never holds.  As a successful example, consider the graph
that is obtained when attempting to show the satisfiability of
$\psi_1\wedge EG\, \neg q$,
(where $\Gamma':=\{\neg q,\Diamond EG\, \neg q\}$):\\[-20pt]
\begin{small}
\begin{prooftree}
\alwaysRootAtTop
\AxiomC{$\Gamma_5$}
\LeftLabel{\scriptsize$(\Diamond)$}
\UnaryInfC{$\Gamma'$}
\LeftLabel{\scriptsize$(\wedge,\nu)$}
\UnaryInfC{$EG\,\neg q=:\Gamma_5$}
\LeftLabel{\scriptsize$(\Diamond)$}
\UnaryInfC{$\Gamma',p,r$}
\AxiomC{}
\LeftLabel{\scriptsize$(\lightning)$}
\UnaryInfC{$\Gamma',q,r\vee\Box \psi_1$}
\AxiomC{$\Gamma_4$}
\RightLabel{\scriptsize$(\Diamond)$}
\UnaryInfC{$\Gamma',\neg p,\Box \psi_2[X\mapsto \psi_1]$}
\RightLabel{\scriptsize$(\vee,\wedge,\mu)^*$}
\BinaryInfC{$\psi_2[X\mapsto \psi_1],EG\,\neg q=:\Gamma_4$}
\RightLabel{\scriptsize$(\Diamond)$}
\UnaryInfC{$\Gamma',p,\Box \psi_2[X\mapsto \psi_1]$}
\LeftLabel{\scriptsize$(\vee)$}
\BinaryInfC{$\Gamma',p,r\vee\Box \psi_2[X\mapsto \psi_1]$}
\AxiomC{$\Gamma_3$}
\RightLabel{\scriptsize$(\Diamond)$}
\UnaryInfC{$\Gamma',\Box \psi_1$}
\LeftLabel{\scriptsize$(\vee,\wedge,\mu,\nu)^*$}
\BinaryInfC{$\psi_2,EG\,\neg q=:\Gamma_3$}
\LeftLabel{\scriptsize$(\wedge)$}
\UnaryInfC{$\psi_2\wedge EG\,\neg q$}
\end{prooftree}
\end{small}
The two loops through $\Gamma_3$ and $\Gamma_4$ are unsuccessful as
they indefinitely postpone the satisfaction of the deferrals
$\delta_2$ and $\delta_3$, respectively; also there is the
unsuccessful clashing node $\Gamma',q,r\vee\Box \psi_1$, containing
both $q$ and $\neg q$. However, the loop through $\Gamma_5$ is
successful since it contains no deferral that is never finished; as
all branching in this example is disjunctive, the single successful
loop suffices to show that the initial node is successful.  Our
algorithm implements this check for `good' and `bad' loops by
\emph{simultaneously} tracking all deferrals that occur through the
proof graph, checking whether each deferral is eventually finished.

\end{example}

\Ni We fix an input formula $\psi_0$ and denote the Fischer-Ladner
closure~\cite{Kozen88} of $\psi_0$ by $\FLphi$; notice that
$|\FLphi|\leq |\psi_0|$.  Let $\nodes = \Pow(\FLphi)$ be the set of
all \emph{nodes} and $\snodes\subseteq\nodes$ the set of all
\emph{state nodes}, i.e.\ nodes that contain only $\top$, non-clashing
propositional literals (where $p$ \emph{clashes} with $\neg p$) and modal
literals; so $|\snodes|\leq|\nodes|\leq 2^{|\psi_0|}$.  Put
\begin{equation*}
\fnodes=\{(\Gamma,d)\in\nodes\times\Pow(\FLphi)\mid d\subseteq \Gamma\},
\quad\text{and}\quad
\fnodes_G=\{(\Gamma,d)\in\fnodes\mid\Gamma\in G\}\text{ for $G\subseteq\nodes$},
\end{equation*}
recalling that nodes are just sets of formulas; note
$|\fnodes|\leq 3^{|\psi_0|}$. Elements $v=(\Gamma,d)\in\fnodes$ are
called \emph{focused nodes}, with \emph{label} $l(v)=\Gamma$ and
\emph{focus} $d$. The idea of focusing single eventualities comes from
work on LTL and CTL~\cite{LangeStirling01,BruennlerLange08}. In the
alternation-free $\mu$-calculus, eventualities may give rise to
multiple deferrals so that one needs to focus \emph{sets of deferrals}
instead of single eventualities. Our algorithm incrementally builds a
set of nodes but performs fixpoint computations on $\Pow(\fnodes)$,
essentially computing winning regions of the corresponding B\"uchi
game (with the target set of player 0 being the nodes with empty
focus) on-the-fly.

\begin{definition}[Conclusions]\upshape
  For a node $\Gamma\in \nodes$ and a set $\mathcal{S}$ of tableau rules, the
  set of \emph{conclusions} of $\Gamma$ under $\mathcal{S}$ is
  \begin{equation*}
    \concl(\mathcal{S},\Gamma)=\{\{\Gamma_1,\ldots,\Gamma_n\}\in
    \Pow(\nodes)\mid (\Gamma/\Gamma_1 \ldots \Gamma_n)\in\mathcal{S}\}.
  \end{equation*}	
  We define $\concl(\Gamma)$ as $\concl(\mathcal{R}_m,\Gamma)$ if $\Gamma$ is a state
  node and as $\concl(\mathcal{R}_p,\Gamma)$ otherwise.  A set
  $N\subseteq \nodes$ of nodes is \emph{fully expanded} if for
  each $\Gamma\in N$, $\bigcup\concl(\Gamma)\subseteq N$.
\end{definition}

\begin{definition}[Deferral tracking]\upshape
  Given a node $\Gamma = \psi_1,\ldots,\psi_n,\phi$
	and a state node 
	$\Delta \in \snodes$ that contains $[a]\psi_1,\ldots,[a]\psi_n,\langle a\rangle\phi$
	as a subset, we say that $\Gamma$ \emph{inherits $\phi$ from $(\langle a\rangle\phi,\Delta)$} 
	and $\psi_i$ from $([a]\psi_i,\Delta)$.
  For a non-state node $\Delta\in \nodes$, a node $\Gamma\in\nodes$ with $\phi\in \Gamma$, and $\psi\in \Delta$,
  $\Gamma$ \emph{inherits $\phi$ from $(\psi,\Delta)$} if $\Gamma=\Gamma_i$ is conclusion of a non-modal rule
  $(\Gamma_0/\Gamma_1 \ldots \Gamma_n)$ with $\Gamma_0=\Delta$
	and either $\psi$ has one of the
  forms $\phi$, $\phi\lor\chi$, $\chi\lor\phi$,
  $\phi\land\chi$, $\chi\land\phi$, or $\psi=\eta X.\;\chi$ and
  $\phi=\chi[X\mapsto\psi]$.
  We put
\begin{align*}
\Rec_m(\phi,\langle a\rangle\phi,\Delta) & =
\{\Gamma\in \nodes\mid \Gamma\text{ inherits $\phi$ from 
$(\langle a\rangle\phi,\Delta)$}\}\\
\Rec_m(\phi,[a]\phi,\Delta) & =
\{\Gamma\in \nodes\mid \Gamma\text{ inherits $\phi$ from 
$([a]\phi,\Delta)$}\}\\
\Rec_p(\phi,\psi,\Delta)&=
\{\Gamma\in \nodes\mid \Gamma\text{ inherits $\phi$ from 
$(\psi,\Delta)$}\},
\end{align*}
where $\Delta$ is a state node in the first two clauses and a non-state node in
the third clause. We write $\mathit{evs}$ for the set of eventualities in $\FLphi$. For
a node $\Gamma\in \nodes$, the set of deferrals of $\Gamma$ is
\begin{align*}
d(\Gamma)=\{\delta\in \Gamma\mid
\exists \theta\in\mathit{evs}.\,\delta\in\mathit{dfr}(\theta)\}.
\end{align*}
For a set $d\neq\emptyset$ of deferrals and nodes $\Gamma,\Delta\in\nodes$, we put
\begin{tabbing}
$d_{\Delta\leadsto \Gamma}=\{\delta\in d(\Gamma)\mid \;\exists \theta\in \mathit{evs}.\,$\=$\exists\langle a\rangle\delta\in d.\;\Gamma\in\Rec_m(\delta,\langle a\rangle\delta,\Delta)$ and $\delta,\langle a\rangle\delta\in\mathit{dfr}(\theta)$ or \\
\>$\exists [a]\delta\in d.\;\Gamma\in\Rec_m(\delta,[a]\delta,\Delta)$
and $\delta,\langle a\rangle\delta\in\mathit{dfr}(\theta)  \}$
\end{tabbing}
if $\Delta$ is a state node, and
\begin{align*}
d_{\Delta\leadsto \Gamma}=\{& \delta_1\in d(\Gamma)\mid \exists \theta\in\mathit{evs}.
\exists \delta_2\in d.\;\Gamma\in\Rec_p(\delta_1,\delta_2,\Delta)\text{ and }
\delta_1,\delta_2\in \mathit{dfr}(\theta)\}
\end{align*}
if $\Delta$ is a non-state node.  I.e.~$d_{\Delta\leadsto \Gamma}$ is the set of
deferrals that is obtained by \emph{tracking} $d$ from $\Delta$ to $\Gamma$,
where $\Gamma$ is the conclusion of a rule application to $\Delta$. We put
$\emptyset_{\Delta\leadsto \Gamma}=d(\Gamma)$, with the intuition that if the focus
$d$ is empty at $(\Delta,d)$, then we \emph{refocus}, i.e.\ choose as new
focus for the conclusion $\Gamma$ the set $d(\Gamma)$ of \emph{all} deferrals in $\Gamma$.


\end{definition}

\begin{example}
  Revisiting the proof graphs from Example~\ref{example:proofs}, we
  fix additional abbreviations
  $\Gamma_6:=\Gamma,\neg p,\Box \psi_2[X\mapsto \psi_1]$,
  $\Gamma_7:=\Gamma',p,r\vee\Box \psi_2[X\mapsto \psi_1]$ and
  $\Gamma_8:=\Gamma',p,r$.  In the first graph, e.g.\
  $d(\Gamma_6) = \{\delta_3\}$ and $d(\Gamma_2) = \{\delta_4\}$; in
  the second graph, e.g.\
  $d(\Gamma_7) = \{r\vee \Box \psi_2[X\mapsto \psi_1]\}$ and
  $d(\Gamma_8)=\emptyset$.  In the first graph, the node $\Gamma_6$
  inherits the deferral $\delta_3$ from $\delta_4$ at $\Gamma_2$,
  i.e.\
  $d(\Gamma_2)_{\Gamma_2\rightsquigarrow\Gamma_6}
  =\{\delta_4\}_{\Gamma_2\rightsquigarrow\Gamma_6}=\{\delta_3\}$
  since
  $\Gamma_6\in \Rec_m(\psi_2[X\mapsto \psi_1],\Box \psi_2[X\mapsto
  \psi_1],\Gamma_2)$.
  Regarding the second graph, $\Gamma_8$ does not inherit any deferral
  from $\Gamma_7$, i.e.\
  $d(\Gamma_7)_{\Gamma_8\rightsquigarrow\Gamma_7} = \{r\vee \Box
  \psi_2[X\mapsto \psi_1]\}_{\Gamma_8\rightsquigarrow\Gamma_7} =
  \emptyset$
  since
  $\Gamma_8\in \Rec_p(r,r\vee\Box \psi_2[X\mapsto \psi_1],\Gamma_7)$
  but $r\vee\Box \psi_2[X\mapsto \psi_1]\in\mathit{dfr}(\psi_1)$ while
  $r\notin\mathit{dfr}(\psi_1)$, i.e.\
  $r\vee\Box \psi_2[X\mapsto \psi_1]$ belongs to $\psi_1$ but $r$ does
  not.  This corresponds to the intuition that $\Gamma_8$ represents a
  branch originating from $\Gamma_7$ that actually finishes the
  deferral $r\vee\Box \psi_2[X\mapsto \psi_1]$.
\end{example}
We next introduce the functionals underlying the fixpoint computations for
propagation of satisfiability and unsatisfiability.
\begin{definition}\upshape
Let $C\subseteq\fnodes$ be a set of focused nodes. 
We define the functions $f:\Pow(C)\to\Pow(C)$ and $g:\Pow(C)\to\Pow(C)$ by
\begin{align*}
  f(Y)&=\{(\Delta,d)\in C\mid \forall \Sigma\in 
                      \concl(\Delta).\,\exists  
                      \Gamma\in\Sigma.\,(\Gamma,
											d_{\Delta \leadsto \Gamma})\in Y\}\\
  g(Y)&=\{(\Delta,d)\in C\mid \exists \Sigma\in 
                      \concl(\Delta).\,\forall 
                      \Gamma\in\Sigma.\,(\Gamma,
											d_{\Delta\leadsto \Gamma})\in Y\}
\end{align*}
for $Y\subseteq C$. We refer to $C$ as the \emph{base set} of $f$ and $g$.
\end{definition}
That is, a focused node $(\Delta,d)$ is in $f(Y)$ if each rule
matching $\Delta$ has a conclusion $\Gamma$ \sut $(\Gamma,d')\in Y$,
where the focus $d'$ is the set of deferrals obtained by tracking $d$
from $\Delta$ to~$\Gamma$.

\begin{definition}[Proof transitionals]\upshape
  For $X\subseteq C\subseteq\fnodes$, we define the \emph{proof
    transitionals} $\hat{f}_X:\Pow(C)\to\Pow(C)$,
  $\hat{g}_X:\Pow(C)\to\Pow(C)$ by
\begin{align*}
\hat{f}_{X}(Y)&:=(f(Y)\cap \overline{F})\cup (f(X)\cap F)=f(Y)\cup (f(X)\cap F)\\
\hat{g}_{X}(Y)&:=(g(Y)\cup F)\cap (g(X)\cup\overline{F})=g(X)\cup (g(Y)\cap \overline{F}),
\end{align*}
for $Y\subseteq C$, where $F=\{(\Gamma,d)\in C\mid d=\emptyset\}$ and
$\overline{F}=\{(\Gamma,d)\in C\mid d\neq\emptyset\}$ are the sets of
focused nodes with empty and non-empty focus, respectively, and where $C$ is the base set of $f$ and $g$.
\end{definition}
That is, $\hat{f}_X(Y)$ contains nodes with non-empty focus that have
for each matching rule a successor node in $Y$ as well as nodes
with empty focus that have for each matching rule a successor node in
$X$. The least fixpoint of $\hat{f}_X$ thus consists of those nodes
that finish their focus -- by eventually reaching nodes from $F$ with empty
focus -- and loop to $X$ afterwards.

\begin{lemma}
  The proof transitionals are monotone w.r.t. set inclusion, i.e.\ if
  $X'\subseteq X$, $Y'\subseteq Y$, then
  $\hat{f}_{X'}(Y')\subseteq \hat{f}_{X}(Y)$ and
  $\hat{g}_{X'}(Y')\subseteq \hat{g}_{X}(Y)$.
\label{lemma:aftrack}
\end{lemma}

\begin{definition}[Propagation]\label{def:afpropagation}
  For $G\subseteq \nodes$, we define
  $E_G,A_G\subseteq \fnodes_G$ as
  \begin{equation*}
    E_G=\nu X.\mu Y.\;\hat{f}_X(Y)\quad\text{and}\quad 
		A_G=\mu X.\nu Y.\;\hat{g}_X(Y),
  \end{equation*}
	where $\fnodes_G$ is the base set of $f$ and $g$.
\end{definition}
Notice that in terms of games, the computation of $E_G$ and $A_G$
corresponds to solving an incomplete B\"uchi game.  The set $E_G$
contains nodes $(\Gamma,d)$ for which player $0$ has a strategy to
enforce -- for each infinite play starting at $(\Gamma,d)$ -- the
B\"uchi condition that nodes in $F$, i.e.\ with empty focus, are
visited infinitely often; similarly $A_G$ is the winning region of
player $1$ in the corresponding game, i.e.\ contains the nodes for
which player $1$ has a strategy to enforce an infinite play that
passes $F$ only finitely often or a finite play that gets stuck in a
winning position for player $1$.

\begin{example}
  Returning to Example~\ref{example:proofs}, we have
  $(\Gamma_1,d(\Gamma_1))=(\Gamma_1,\{\psi_1\})\in A_{G_1}$ and
  $(\Gamma_3,d(\Gamma_3))=(\Gamma_3,\{\psi_1\})\in E_{G_2}$ where
  $G_1$ and $G_2$ denote the set of all nodes of the first and the
  second proof graph, respectively; the global caching algorithm
  described later will therefore answer `unsatisfiable' to $\Gamma_1$,
  and `satisfiable' to $\Gamma_3$. To see
  $(\Gamma_1,\{\psi_1\})\in A_{G_1}$ note that
  $A_{G_1}=\nu Y.\, \hat{g}_{A_{G_1}}(Y)$ by definition, so
  $A_{G_1}=(\hat{g}_{A_{G_1}})^n(\fnodes_{G_1})$ for some $n$. For
  each focused node $(\Delta,d)\in \fnodes_{G_1}$ there is a rule
  matching $\Delta$ all whose conclusions $\Gamma$ satisfy
  $(\Gamma,d_{\Delta\rightsquigarrow \Gamma})\in \fnodes_{G_1}$, i.e.\
  $g(\fnodes_{G_1})=\fnodes_{G_1}$. Moreover, since all loops in $G_1$
  indefinitely postpone some eventuality, no node with non-empty focus
  ever reaches one with empty focus, so
  $\hat{g}_{\emptyset}(\fnodes_{G_1})=\overline{F}$. Since $\hat{g}$
  is monotone and $(\Gamma_1,\{\psi_1\})\in \overline{F}$, we obtain
  by induction over $n$ that
  $(\Gamma_1,\{\psi_1\})\in(\hat{g}_{A_{G_1}})^n(\fnodes_{G_1})$.  To
  see $(\Gamma_3,d(\Gamma_3))=(\Gamma_3,\{\psi_1\})\in E_{G_2}$, note
  that that starting from $\Gamma_3$, the single deferral $\psi_1$ can
  be finished in finite time while staying in $E_{G_2}$.  This holds
  because we can reach $(\Gamma_8,\emptyset)$ by branching to the left
  twice and $(\Gamma_8,\emptyset)\in E_{G_2}$, since the loop through
  $\Gamma_5$ does not contain any deferrals whose satisfaction is
  postponed indefinitely and hence is contained in~$E_{G_2}$.
\end{example}

\begin{lemma}
If $G'\subseteq G$, then $E_{G'}\subseteq E_{G}$ and $A_{G'}
\subseteq A_{G}$.
\label{lemma:afsuccunsuccmon}
\end{lemma}

\begin{lemma}
  Let $G\subseteq \nodes$ be fully expanded.  Then
  $E_{G}=\overline{A_{G}}$.
\label{lemma:afcomplementary}
\end{lemma}



\noindent Our algorithm constructs a partial tableau, maintaining sets
$G,U\subseteq\nodes$ of \emph{expanded} and \emph{unexpanded} nodes,
respectively. It computes $E_G,A_G\subseteq \fnodes_G$ in the
propagation steps; as these sets grow monotonically, they can be
computed incrementally.
\begin{alg}[Global caching]\label{alg:gc}
  Decide satisfiability of a closed formula $\phi_0$.
\begin{enumerate}
\item (Initialization) Let $G:=\emptyset$, $\Gamma_0:=\{\phi_0\}$, $U:=\{\Gamma_0\}$.
\item\label{step:afexpand} (Expansion) Pick $t\in U$ and let
  $G:=G\cup \{t\}$, $U:=(U - \{t\})\cup (\bigcup\concl(t) - G)$.
\item\label{step:afpropagate} (Intermediate propagation) Optional:
  Compute $E_G$ and/or $A_G$. If $(\Gamma_0,d(\Gamma_0)) \in E_G$, return `Yes'.  If
  $(\Gamma_0,d(\Gamma_0))\in A_G$, return `No'.
\item If $U\neq \emptyset$, continue with Step~\ref{step:afexpand}. 
\item\label{step:affinal} (Final propagation) Compute $E_G$. If
  $(\Gamma_0,d(\Gamma_0))\in E_G$, return `Yes', else `No'.
\end{enumerate}
\label{alg:afgc}
\end{alg}
Note that in Step~\ref{step:affinal}, $G$ is fully expanded. For
purposes of the soundness proof, we note an immediate consequence of
Lemmas~\ref{lemma:afsuccunsuccmon} and~\ref{lemma:afcomplementary}:
\begin{lemma}
  If some run of the algorithm without intermediate propagation steps
  is successful on input $\phi_0$, then all runs on input $\phi_0$ are
  successful.
\label{lemma:afnonopt}
\end{lemma}

\begin{remark}
  For alternation-free fixpoint logics, the game-based approach
  (e.g.\ \cite{FriedmannLatteLange13}) is to (1.) define a
  nondeterministic co-B\"uchi automaton of size $\mathcal{O}(n)$ that
  recognizes unsuccessful branches of the tableau. This automaton is
  then (2.)  determinized to a deterministic co-B\"uchi automaton of
  size $2^{\mathcal{O}(n)}$ (avoiding the Safra construction using
  instead the method of~\cite{MiyanoHayashi1984}; here,
  alternation-freeness is crucial) and (3.) complemented to a
  deterministic B\"uchi automaton of the same size that recognizes
  successful branches of the tableau. A B\"uchi game is (4.)
  constructed as the product game of the carrier of the tableau and
  the carrier of the B\"uchi automaton. This game is of size
  $2^{\mathcal{O}(n)}$ and can be (5.)  solved in time
  $2^{\mathcal{O}(n)}$. 

  Our global caching algorithm integrates analogues of items (1.) to
  (5.) in one go: We directly construct the B\"uchi game (thus
  replacing (1.) through (4.) by a single definition) step-by-step
  during the computation of the sets $E$ and $A$ of (un)successful
  nodes as nested fixpoints of the proof transitionals; 
	the propagation step corresponds to~(5.).  
	Our algorithm allows for intermediate propagation,
  corresponding to solving the B\"uchi game on-the-fly, i.e.\ before it
  has been fully constructed.  

\end{remark}

\section{Soundness, Completeness and Complexity}\label{section:soundcomp}

\subparagraph*{Soundness} Let $\phi_0$ be a satisfiable formula. By
Lemma~\ref{lemma:afnonopt}, it suffices to show that a run without
intermediate propagation is successful. 

\begin{definition}\upshape
  For a formula $\psi$, we define 
	$\psi_X(\phi)=\psi[X\mapsto\phi]$,
	$\psi_{X}^0=\bot$ and $\psi^{n+1}_{X}=\psi_{X}(\psi^n_{X})$.
        We say that a Kripke structure $\mathcal{K}$ is
        \emph{stabilizing} if for each state $x$ in $\mathcal{K}$,
        each $\mu X.\,\psi$, and each fixpoint-free context $c(-)$
        such that $x\models c(\mu X.\,\psi)$, there is $n\ge 0$ \sut
        $x\models c(\psi^n_{X})$.
\end{definition}

\noindent We note that finite Kripke structures are stabilizing and
import the finite model property (without requiring a bound on model
size) for the $\mu$-calculus from~\cite{Kozen88}; for the rest of the
section, we thus fix w.l.o.g. a stabilizing Kripke structure
$\mathcal K=(W,(R_a)_{a\in A},\pi)$ satisfying the target formula
$\phi_0$ in some state.
\begin{definition}[Unfolding tree]\upshape
  Given a formula $\psi$, an \emph{unfolding tree} $t$ for $\psi$
  consists of the syntax tree of $\psi$ together with a natural number
  as additional label for each node that represents a least fixpoint
  operator.  We denote this number by $t(\kappa,\mu X.\,\phi)$ for an
  occurrence of a fixpoint literal $\mu X.\,\phi$ at position
  $\kappa\in\{0,1\}^*$ in $\psi$.  We define the \emph{unfolding}
  $\psi(t)$ of $\psi$ according to an unfolding tree $t$ for $\psi$ by
 \begin{equation*}
   X(t)  = X \qquad
   (\phi_1\wedge\phi_2)(t)  = \phi_1(t_1)\wedge\phi_2(t_2)
   \quad (\mu X.\,\phi_1)(t)  = (\phi_1(t_1))_{X}^{t(\epsilon,\mu X.\,\phi_1)},
\end{equation*}
where $t_i$ is the $i$-th child of the root of $t$, and similar
clauses for $\langle a\rangle$, $[a]$, $\lor$, and $\nu$ as for
$\land$.
\end{definition}
Given a formula $\psi$, we define the order $<_\psi$ on unfolding
trees for $\psi$ by lexically ordering the lists of labels obtained by
pre-order traversal of the syntax tree of $\psi$.

\begin{definition}[Unfolding]\upshape
  The \emph{unfolding} of a formula $\psi$ at a state $x$ with
  $x\models \psi$ is defined as $\mathit{unf}(\psi,x) = \psi(t)$,
  where $t$ is the least unfolding tree for $\psi$ (w.r.t.\ $<_\psi$)
  such that $x\models\psi(t)$ (such a $t$ exists by stabilization).
\end{definition}
Note that in unfoldings, all least fixpoint literals $\mu X.\,\phi$ are
replaced with finite iterates of $\phi$. 
\begin{theorem}[Soundness]
The algorithm returns `Yes' on input $\phi_0$ if $\phi_0$ is satisfiable.
\label{thm:afsatsucc}
\end{theorem}
\begin{proof}(Sketch) We show that any node $(\Gamma,d)$ that is
  constructed by the algorithm and whose label is satisfied at some
  state $x$ in $\mathcal{K}$ is successful, i.e.\ $(\Gamma,d)\in E_G$; the
  proof is by induction over the maximal modal depth of
  $\mathit{unf}(\delta,x)$ for $\delta \in d$.
\end{proof}

\subparagraph*{Completeness} Assume that the algorithm answers `Yes' on
input $\phi_0$, having constructed the set $E:=E_G$ of successful
nodes. Put $D = \{(\Gamma,d)\in E\mid \Gamma\in\snodes\}$; note
$|D|\le |E| \le 3^{|\phi_0|}$.

\begin{definition}[Propositional entailment]\upshape
  For a finite set $\Psi$ of formulas, we write $\bigwedge\Psi$ for
  the conjunction of the elements of $\Psi$.  We say that $\Psi$
  \emph{propositionally entails} a formula $\phi$ (written 
  $\Psi\PLentails \phi$) if $\bigwedge\Psi\rightarrow \phi$ is a
  propositional tautology, where modal literals are treated as
  propositional atoms and fixpoint literals $\eta X.\phi$ are
  unfolded to $\phi({\eta X.\phi})$ (recall that fixpoint operators are
  guarded). 
\end{definition}

\begin{definition}\upshape\label{defn:sufficiency}
  We denote the set of formulas in a node $\Gamma$ that do \emph{not}
  belong to an eventuality $\theta$ by
\begin{align*}
N(\Gamma,\theta) = \{\phi\in \Gamma \mid \phi\notin\mathit{dfr}(\theta)\}.
\end{align*}
\Ni A set $d$ of deferrals is \emph{sufficient} for
$\delta\in\defer(\theta)$ at a node $\Gamma$, in symbols $d\vdash_\Gamma \delta$, if
$d\, \cup N(\Gamma,\theta)\PLentails \delta$.  We write $\vdash_\Gamma \delta$ to
abbreviate $\emptyset\vdash_\Gamma \delta$.
\end{definition}

\begin{definition}[Timed-out tableau]\label{def:aftableau}\upshape
  Let $U\subseteq\snodes\times\snodes$ and let
  $L\subseteq U\times U$. We denote the set of $L$-successors of
  $v\in U$ by $L(v)=\{w\mid (v,w)\in L\}$.  Let $d$ be a set of
  deferrals. We put $\mathit{to}(\emptyset,n)=U$ for all $n$
  ($\mathit{to}$ for \emph{timeout}). For $d\neq \emptyset$, we put
  $\mathit{to}(d,0)=\emptyset$ and define $\mathit{to}(d,m+1)$ to be
  the set of of focused nodes $(\Delta,d')$ \sut writing
  $\concl(\Delta)=\{\Sigma_1,\dots,\Sigma_n\}$, we have
  $L(\Delta,d')=\{(\Gamma_1,d_1),\dots,(\Gamma_n,d_n)\}$ where for each $i$ there
  exists $\Gamma\in\Sigma_i$ \sut
	\begin{itemize}
  \item $\Gamma_i\PLentails\bigwedge\Gamma$ and $d_i\vdash_{\Gamma_i}d'_{\Delta\rightsquigarrow\Gamma}$, and
	\item $(\Gamma_i,d_i)\in \mathit{to}(d'',m)$ for
	some $d''\subseteq d(\Gamma_i)$ with $d''\vdash_{\Gamma_i}d_{\Delta\rightsquigarrow\Gamma}$.
  \end{itemize}
  If for each focused node $(\Gamma,d)\in U$ there is a number $m$ \sut
  $(\Gamma,d)\in \mathit{to}(d(\Gamma),m)$, then $L$ is a \emph{timed-out tableau} over
  $U$.
\end{definition}

\noindent Roughly, $\mathit{to}(d,m)$ can be understood as the set of
all focused nodes in $U$ that finish all deferrals in $d$ within $m$
modal steps, i.e.\ with \emph{time-out} $m$; this is similar to
Kozen's \emph{$\mu$-counters}~\cite{Kozen83}.

\begin{lemma}[Tableau existence]\label{lemm:aftableau-existence}
  There exists a timed-out tableau  over $D$.
\end{lemma}
\begin{proof}[Proof sketch] Since $D\subseteq E_G$, we can define
  $L\subseteq D\times D$ in such a way that all paths in $L$ visit $F$
  (the set of nodes with empty focus) infinitely often, so every
  deferral contained in some node in $D$ will be focused by the
  unavoidable eventual refocusing; this new focus will in turn
  eventually be finished so that $L$ is a timed-out tableau.
\end{proof}

\Ni For the rest of the section, we fix a timed-out tableau $L$ over $D$
and define a Kripke structure $\mathcal{K}=(D,(R_a)_{a\in A},\pi)$
by taking $R_a(v)$ to be the set of focused nodes in $L(v)$ whose label is the
conclusion of an $(\langle a\rangle)$-rule that matches $l(v)$ and by
putting $\pi(p)=\{v\in D\mid p\in l(v)\}$.

\begin{definition}[Pseudo-extension]\upshape
  The \emph{pseudo-extension} $\psem{\phi}$ of $\phi$ in $D$ is
  \begin{equation*}
    \psem{\phi} = \{ v \in D \mid l(v)\PLentails\phi \}.
  \end{equation*}
\end{definition}

\begin{lemma}[Truth]\label{lemma:truth}
  In the Kripke structure $\mathcal{K}$,
  $\psem{\psi}\subseteq\sem{\psi}$ for all $\psi\in\FLphi$.
\end{lemma}

\begin{proof}[Proof sketch] Induction on $\psi$, with an additional
  induction on time-outs in the case for least fixpoint
  literals, exploiting alternation-freeness. 
\end{proof}

\begin{corollary}[Completeness]
\label{cor:completeness}
If a run of the algorithm with input $\phi_0$ returns `Yes',
then $\phi_0$ is satisfiable.
\end{corollary}
\begin{proof}[Proof sketch]
  Combine the existence lemma and the truth lemma to obtain a model
  over~$D$.  Since $(\{\phi_0\},d(\{\phi_0\}))\in E$ and
  $\psem{\phi_0}\subseteq\sem{\phi_0}$, there is a focused node in $D$
  that satisfies $\phi_0$.
\end{proof}


\noindent As a by-product, our model construction yields

\begin{corollary}\label{cor:model-size}
  Every satisfiable alternation-free fixpoint formula $\phi_0$ has a model of
  size at most $3^{|\phi_0|}$.
\end{corollary}

\noindent Thus we recover the bound of $2^{\mathcal{O}(n)}$ for the
alternation-free relational $\mu$-calculus, which can be obtained,
e.g., by carefully adapting results from~\cite{FriedmannLange13a} to
the alternation-free case; for the alternation-free fragment of the
alternating-time $\mu$-calculus, covered by the coalgebraic
generalization discussed next, the best previous bound appears to be
$n^{\lO(n)}=2^{\lO(n\log n)}$~\cite{ScheweThesis}.

\subparagraph*{Complexity}\label{section:complexity}
\noindent Our algorithm has optimal complexity (given that the problem
is known to be \ExpTime-hard):

\begin{theorem}
  The global caching algorithm decides the satisfiability problem of
  the alternation-free $\mu$-calculus in \ExpTime, more precisely in
  time $2^{\lO(n)}$.
\label{prop:exptime}
\end{theorem}

\subparagraph*{The Alternation-Free Coalgebraic $\mu$-Calculus
}\label{section:coalg}

Coalgebraic logic~\cite{CirsteaEA11} serves as
a unifying framework for modal logics beyond standard relational
semantics, subsuming systems with, e.g., probabilistic, weighted,
game-oriented, or preference-based behaviour under the concept of
coalgebras for a set functor $F$. All our results lift to the level of
generality of the (alternation-free) coalgebraic
$\mu$-calculus~\cite{CirsteaEA11a}; details are in a technical report
at \url{https://www8.cs.fau.de/hausmann/afgc.pdf}. In consequence, our
results apply also to the alternation-free fragments of the
alternating-time $\mu$-calculus~\cite{AlurEA02}, probabilistic
fixpoint logics, and the monotone $\mu$-calculus (the ambient fixpoint
logic of Parikh's game logic~\cite{Parikh85}), as all these can be
cast as instances of the coalgebraic $\mu$-calculus.

\section{Implementation and Benchmarking}\label{section:cool}

The global caching algorithm has been implemented as an extension of
the \emph{Coalgebraic Ontology Logic Reasoner}
(COOL)~\cite{GorinEA14}, a generic reasoner for coalgebraic modal
logics, available at
\url{https://www8.cs.fau.de/research:software:cool}. COOL achieves its
genericity by instantiating an abstract core reasoner that works for
all coalgebraic logics to concrete instances of logics; our global
caching algorithm extends this core. Instance logics implemented in
COOL currently include relational, monotone, and alternating-time
logics, as well as any logics that arise as fusions thereof. In
particular, this makes COOL, to our knowledge, the only implemented
reasoner for the alternation-free fragment of the alternating-time
$\mu$-calculus (a tableau calculus for the sublogic ATL is
prototypically implemented in the TATL reasoner~\cite{David13}) and
the star-nesting free fragment of Parikh's game logic.

Although our tool supports the full alternation-free $\mu$-calculus,
we concentrate on CTL for experiments, as this appears to be the only
candidate logic for which substantial sets of benchmark formulas are
available~\cite{GoreEA11}. CTL reasoners can be broadly classified as
being either \emph{top-down}, i.e.\ building graphs or tableaux by
recursion over the formula, or \emph{bottom-up}; the two groups
perform very differently~\cite{GoreEA11}. We compare our implementation
with the top-down solvers TreeTab~\cite{Gore14},
GMUL~\cite{GoreEA11}, MLSolver~\cite{FriedmannLange10MLSolver} and the
bottom-up solvers CTL-RP~\cite{ZhangEA14} and BDDCTL~\cite{GoreEA11}.
Out of the top-down solvers, only TreeTab is single-pass like COOL; however, 
TreeTab has suboptimal (doubly exponential) worst-case runtime. MLSolver
supports the full $\mu$-calculus. For MLSolver, CTL-RP and BDDCTL, formulas
have first been \emph{compacted}~\cite{GoreEA11}. 
All tests have been executed on a system
with Intel Core i7 3.60GHz CPU with 16GB RAM, and a stack limit of 512MB.

On the benchmark formulas of~\cite{GoreEA11}, COOL essentially
performs similarly as the other top-down tools, and closer to the
better tools when substantial differences show up.  As an example, the
runtimes of COOL, TreeTab, GMUL, MLSolver, CTL-RP, and BDDCTL on the
Montali-formulas~\cite{MontaliEA08,GoreEA11} are shown in
Figure~\ref{fig:graphsa}.
To single out one more example, Figure~\ref{fig:graphsb} shows
the runtimes for the alternating bit protocol benchmark from~\cite{GoreEA11};
COOL performs closer to GMUL than
to MLSolverc on these formulas.  

This part of the evaluation
may be summed up as saying that COOL performs well despite being, at
the moment, essentially unoptimized: the only heuristics currently
implemented is a simple-minded dependency of the frequency of
intermediate propagation on the number of unexpanded
nodes. 


\begin{figure*}[ht!]
\captionsetup{justification=centering}
    \centering
    \begin{subfigure}[t]{0.5\textwidth}
        \centering
\begin{tikzpicture}
\begin{semilogyaxis}[
minor tick num=1,
xtick={0,21,41,61,81,101,121,141},
ytick={0.001,0.01,0.1,1,10,100,1000},
yticklabels={$0.001$,$0.01$,$0.1$,$1$,$10$,$100$,$1000$},
every axis y label/.style=
{at={(ticklabel cs:0.5)},rotate=90,anchor=center},
every axis x label/.style=
{at={(ticklabel cs:0.5)},anchor=center},
tiny,
width=7.5cm,
height=6cm,
transpose legend,
legend columns=2,
legend style={at={(0.5,-0.13)},anchor=north},
ymode=log,
xlabel={value of n},
ylabel={runtime (s)},
xmin=0,
xmax=141,
ymin=0.001,
ymax=1000,
legend entries={COOL,TreeTab,GMUL,MLSolverc, CTL-RPc, BDDCTLc}]
\addplot[mark=triangle*,mark options={scale=0.8}]  table {
1 0.003333
6 0.003333
11 0.001
16 0.003333
21 0.003333
26 0.003333
31 0.006666
36 0.013333
41 0.019999
46 0.033333
51 0.029998999999999998
56 0.039999
61 0.056665999999999994
66 1.023333
71 0.5133329999999999
76 0.379999
81 0.179999
86 0.8666659999999999
91 0.42999899999999996
96 3.586666
101 5.659999
106 1.7666659999999998
111 16.716664
116 16.559998
121 26.489997
126 1001
131 9.836665
136 1001
};
\addplot[mark=o,mark options={scale=0.8}] table {
1 0.001
6 0.001
11 0.001
16 0.001
21 0.001
26 0.001
31 0.001
36 0.001
41 0.001
46 0.001
51 0.003333
56 0.009999
61 0.023333
66 0.066666
71 0.16666599999999998
76 0.486666
81 1.473333
86 4.536666
91 17.006664
96 62.61666
101 250.246641
106 1001
};
\addplot[mark=square,mark options={scale=0.8}] table {
1 0.001
6 0.001
11 0.001
16 0.001
21 0.001
26 0.006666
31 0.039999
36 0.36333299999999996
41 3.196666
46 5.316666
51 22.466664
56 92.91999
61 1001
};
\addplot [mark=pentagon, mark options={scale=0.9}] table {
1 0.001
6 0.001
11 0.016666
16 0.33666599999999997
21 8.439999
26 180.009981
31 1001
};
\addplot [mark=diamond, mark options={scale=0.9}] table {
1 0.001
6 0.001
11 0.001
16 0.001
21 0.001
26 0.001
31 0.001
36 0.001
41 0.001
46 0.001
51 0.001
56 0.001
61 0.001
66 0.001
71 0.003333
76 0.003333
81 0.003333
86 0.003333
91 0.003333
96 0.003333
101 0.003333
106 0.003333
111 0.003333
116 0.003333
121 0.003333
126 0.006666
131 0.003333
136 0.006666
141 0.006666
146 0.006666
151 0.009999
156 0.009999
161 0.009999
166 0.009999
171 0.009999
176 0.009999
181 0.013333
186 0.013333
191 0.013333
196 0.013333
201 0.016666
206 0.016666
211 0.013333
216 0.016666
221 0.016666
226 0.019999
231 0.019999
236 0.019999
241 0.023333
246 0.023333
251 0.023333
256 0.026666
261 0.026666
266 0.026666
271 0.029998999999999998
276 0.029998999999999998
281 0.033333
286 0.033333
291 0.036666
296 0.036666
};
\addplot [mark=star, mark options={scale=0.9}] table {
1 0.001
6 0.001
11 0.001
16 0.001
21 0.001
26 0.001
31 0.001
36 0.001
41 0.001
46 0.001
51 0.001
56 0.001
61 0.001
66 0.003333
71 0.003333
76 0.003333
81 0.003333
86 0.003333
91 0.003333
96 0.006666
101 0.006666
106 0.009999
111 0.009999
116 0.013333
121 0.013333
126 0.013333
131 0.016666
136 0.016666
141 0.016666
146 0.019999
151 0.023333
156 0.023333
161 0.026666
166 0.029998999999999998
171 0.033333
176 0.036666
181 0.036666
186 0.036666
191 0.043332999999999997
196 0.046666
201 0.049998999999999995
206 0.053333
211 0.059999
216 0.066666
221 0.06999899999999999
226 0.06999899999999999
231 0.079999
236 0.079999
241 0.08666599999999999
246 0.093333
251 0.093333
256 0.103333
261 0.103333
266 0.109999
271 0.11666599999999999
276 0.123333
281 0.129999
286 0.13666599999999998
291 0.143333
296 1001
};

\end{semilogyaxis}
\end{tikzpicture}\centering
        \subcaption{Montali, $n=1$ \quad(satisfiable)}
    \end{subfigure}%
    ~ 
    \begin{subfigure}[t]{0.5\textwidth}
        \centering
\begin{tikzpicture}
  \begin{semilogyaxis}[
    minor tick num=1,
    xtick={0,21,41,61,81},
    ytick={0.001,0.01,0.1,1,10,100,1000},
    yticklabels={$0.001$,$0.01$,$0.1$,$1$,$10$,$100$,$1000$},
    every axis y label/.style={at={(ticklabel cs:0.5)},rotate=90,anchor=center},
    every axis x label/.style={at={(ticklabel cs:0.5)},anchor=center},
    tiny,
    width=7.5cm,
    height=6cm,
    transpose legend,
    legend columns=2,
    legend style={at={(0.5,-0.13)},anchor=north},
    ymode=log,
    xlabel={value of n},
    ylabel={runtime (s)},
    xmin=0,
    xmax=81,
    ymin=0.001,
    ymax=1000,
    legend entries={COOL,TreeTab,GMUL,MLSolverc, CTL-RPc, BDDCTLc}]
    \addplot[mark=triangle*,mark options={scale=0.8}] table {
1 0.003333
6 0.003333
11 0.003333
16 0.003333
21 0.013333
26 0.056665999999999994
31 0.346666
36 3.886666
41 26.63333
46 214.999978
51 966.219903
56 1001
};
    \addplot[mark=o,mark options={scale=0.8}] table {
1 0.001
6 0.001
11 0.001
16 0.001
21 0.001
26 0.006666
31 0.029998999999999998
36 0.259999
41 1.386666
46 7.396665
51 32.999996
56 163.33665
61 1001
};
    \addplot[mark=square,mark options={scale=0.8}] table {
1 0.001
6 0.001
11 0.001
16 0.001
21 0.001
26 0.001
31 0.001
36 0.006666
41 0.033333
46 0.13999899999999998
51 0.5899989999999999
56 2.619999
61 13.986665
66 104.059989
71 1001
};
    \addplot[mark=pentagon,mark options={scale=0.8}] table {
1 0.001
6 0.001
11 0.003333
16 0.019999
21 0.19666599999999998
26 2.449999
31 31.749996
36 439.193289
41 1001
};
    \addplot[mark=diamond,mark options={scale=0.8}] table {
1 0.001
6 0.003333
11 0.006666
16 0.029998999999999998
21 0.09999899999999999
26 0.529999
31 3.553332
36 34.289996
41 486.693284
46 1001
};
    \addplot[mark=star,mark options={scale=0.8}] table {
1 0.001
6 0.001
11 0.001
16 0.001
21 0.001
26 0.003333
31 0.003333
36 0.006666
41 0.009999
46 0.013333
51 0.019999
56 0.026666
61 0.049998999999999995
66 0.073333
71 0.106666
76 0.159999
81 0.23333299999999998
86 0.269999
91 0.35999899999999996
96 0.49333299999999997
101 0.629999
106 0.849999
111 1.109999
116 1.129999
121 1.386666
126 1.693333
131 2.046666
136 2.459999
141 2.886666
146 3.453332
151 4.066666
156 4.879999
161 5.773332
166 5.466666
};
  \end{semilogyaxis}
\end{tikzpicture}
        \subcaption{Montali, $n=1$ \quad(unsatisfiable)}
    \end{subfigure}
    \caption{Runtimes for the Montali-formulas}
    \label{fig:graphsa}
\end{figure*}
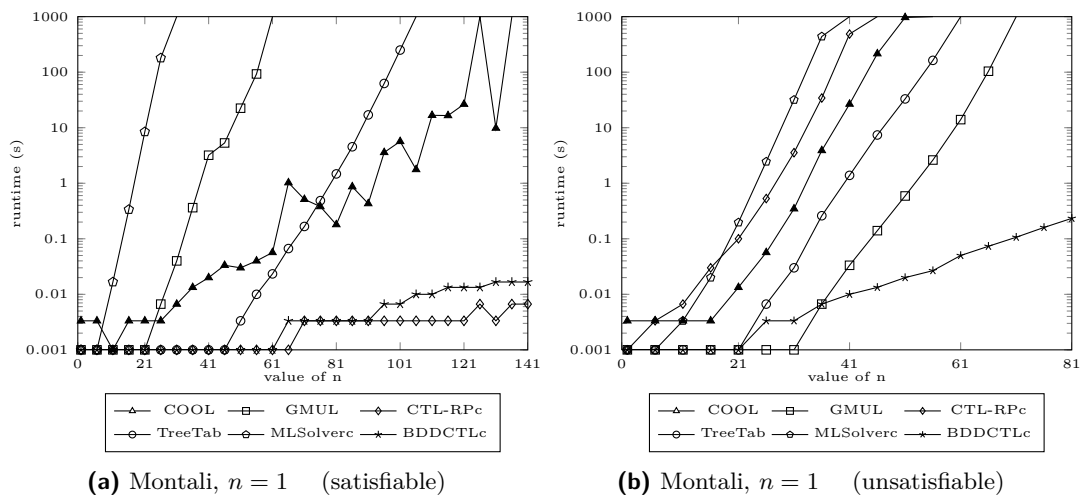

\begin{figure*}[ht!]
\captionsetup{justification=centering}
    \centering
\begin{tabular}{|l | l l l l l l|}
\hline
Type of formula & COOL & TreeTab & GMUL & MLSolverc & BDDCTLc & CTL-RPc\\
\hline
(i) & <0.01 & <0.01 & <0.01 & 0.02 & <0.01 & 0.02\\
(ii) & 0.12 &  -- & 0.02 & 0.95 & <0.01 & 0.15\\
(iii) & 0.12 & -- & 0.02 & 0.87 & <0.01 & 0.16\\
\hline
\end{tabular}
    \caption{Runtimes (in s) for the Alternating Bit Protocol formulas}
    \label{fig:graphsb}
\end{figure*}

\noindent In addition, we design two series of unsatisfiable benchmark formulas
that have an exponentially large search space but allow for detection
of unsatisfiability at an early stage. Recall that in CTL we can
express the statement `in the next step, the $n$-bit counter $x$
represented by the variables $x_1,\dots,x_n$ will be incremented'
(with wraparound) as a formula $c(x,n)$ of polynomial size in $n$.  We
define unsatisfiable formulas $\mathit{early}(n,j,k)$ that specify an $n$-bit
counter $p$ with $n$ bits and additionally branch after $2^j$
steps (i.e.\ when $p_j$ holds) to start a counter $r$ with $k$ bits
which in turn forever postpones the eventuality $EF\;p$:
\begin{align*}
\mathit{early}(n,j,k) = & \;\mathit{start}_p \wedge \mathit{init}(p,n) \wedge \mathit{init}(r,k) \wedge AG\;((r \to c(r,k))\wedge (p \to c(p,n))) \wedge\\
               & \; AG\;((\textstyle \bigwedge_{0\leq i\leq j} p_i \to EX(\mathit{start}_r \wedge EF\;p)) \wedge \neg(p \wedge r)\wedge (r\to AX\; r))\\
\mathit{init}(x,m) = &\; AG\; ((\mathit{start}_x \to (x \wedge\textstyle \bigwedge_{0\leq i<m} \neg x_i))\wedge (x\to EX\; x)).
\end{align*}
Note here that $\mathit{init}$ uses $x$ as a string argument;
$\mathit{start}_x$ is an atom indicating the start of counter $x$, and
the atom $x$ itself indicates that the counter $x$ is running. The
second series of unsatisfiable formulas $\mathit{early}_{gc}(n,j,k)$
is obtained by extending the formulas $\mathit{early}(n,j,k)$ with the
additional requirement that a further counter $q$ with $n$ bits is
started infinitely often, but at most at every second
step: 
\begin{align*}
\mathit{early}_{gc}(n,j,k) = & \;\mathit{early}(n,j,k) \wedge \mathit{b} \wedge \mathit{init}(q,n) \wedge AG\;(\neg(p \wedge q)\wedge \neg(q \wedge r)  \wedge (q \to c(q,n)))\\
        & \wedge\;AG\;(AF\;\mathit{b} \wedge (\mathit{b} \to (EX\;p \wedge EX\;\mathit{start}_q 
                               \wedge AX\;\neg \mathit{b})))
\end{align*}
\vspace{-2em}

\begin{figure*}[ht!]
\captionsetup{justification=centering}
    \centering
    \begin{subfigure}[t]{0.5\textwidth}
        \centering
\begin{tikzpicture}
\begin{semilogyaxis}[
minor tick num=1,
xtick={0,2,4,6,8,10,12,14,16,18,20},
ytick={0.001,0.01,0.1,1,10,100,1000},
yticklabels={$0.001$,$0.01$,$0.1$,$1$,$10$,$100$,$1000$},
every axis y label/.style=
{at={(ticklabel cs:0.5)},rotate=90,anchor=center},
every axis x label/.style=
{at={(ticklabel cs:0.5)},anchor=center},
tiny,
width=7.5cm,
height=6cm,
transpose legend,
legend columns=2,
legend style={at={(0.5,-0.13)},anchor=north},
ymode=log,
xlabel={value of n},
ylabel={runtime (s)},
xmin=0,
xmax=20,
ymin=0.001,
ymax=1000,
legend entries={COOL,TreeTab,GMUL,MLSolverc, CTL-RPc, BDDCTLc}]
\addplot[mark=triangle*,mark options={scale=0.8}]  table {
1 0.003333
2 0.003333
3 0.006666
4 0.006666
5 0.013333
6 0.033333
7 0.019999
8 0.026666
9 0.029998999999999998
10 0.026666
11 0.043332999999999997
12 0.056665999999999994
13 0.056665999999999994
14 0.046666
15 0.043332999999999997
16 0.049998999999999995
17 0.049998999999999995
18 0.056665999999999994
19 0.056665999999999994
20 0.09999899999999999
};
\addplot[mark=o,mark options={scale=0.8}] table {
1 0.001
2 0.001
3 0.001
4 0.001
5 0.001
6 0.001
7 0.001
8 0.001
9 0.003333
10 0.006666
11 0.013333
12 0.029998999999999998
13 0.053333
14 0.123333
15 0.253333
16 0.509999
17 0.799999
18 1.659999
19 2.923333
20 5.366666
};
\addplot[mark=square,mark options={scale=0.8}] table {
1 0.001
2 0.001
3 0.001
4 0.003333
5 0.009999
6 0.029998999999999998
7 0.059999
8 0.159999
9 0.48333299999999996
10 1.906666
11 6.473332
12 13.973331
13 35.276663
14 78.499992
15 164.873316
16 340.633299
17 653.876601
18 1001
};
\addplot [mark=pentagon, mark options={scale=0.9}] table {
1 0.013333
2 0.22666599999999998
3 6.279999
4 201.783313
5 1001
};
\addplot [mark=diamond, mark options={scale=0.9}] table {
1 0.003333
2 0.006666
3 0.156666
4 170.123316
5 1001
};
\addplot [mark=star, mark options={scale=0.9}] table {
1 0.001
2 0.003333
3 0.013333
4 0.096666
5 0.859999
6 5.483332
7 28.066663
8 135.433319
9 670.226599
10 1001
};

\end{semilogyaxis}
\end{tikzpicture}\centering
        \subcaption{$\mathit{early}(n,4,2)$ \quad(unsatisfiable)}
    \end{subfigure}%
    ~ 
    \begin{subfigure}[t]{0.5\textwidth}
        \centering
\begin{tikzpicture}

\begin{semilogyaxis}[
minor tick num=1,
xtick={0,2,4,6,8,10,12,14,16,18,20},
ytick={0.001,0.01,0.1,1,10,100,1000},
yticklabels={$0.001$,$0.01$,$0.1$,$1$,$10$,$100$,$1000$},
every axis y label/.style=
{at={(ticklabel cs:0.5)},rotate=90,anchor=center},
every axis x label/.style=
{at={(ticklabel cs:0.5)},anchor=center},
tiny,
width=7.5cm,
height=6cm,
transpose legend,
legend columns=2,
legend style={at={(0.5,-0.13)},anchor=north},
ymode=log,
xlabel={value of n},
ylabel={runtime (s)},
xmin=0,
xmax=20,
ymin=0.001,
ymax=1000,
legend entries={COOL,TreeTab,GMUL,MLSolverc,CTL-RPc,BDDCTLc}]

\addplot[mark=triangle*,mark options={scale=0.8}]  table {
1 0.006666
2 0.023333
3 0.043332999999999997
4 0.126666
5 0.306666
6 0.353333
7 0.39666599999999996
8 0.843333
9 0.7266659999999999
10 2.829999
11 3.8833320000000002
12 2.666666
13 4.036666
14 4.449999
15 1.936666
16 2.149999
17 2.213333
18 17.833331
19 17.646664
20 4.259999
};
\addplot[mark=o,mark options={scale=0.8}] table {
1 0.001
2 0.003333
3 0.006666
4 0.029998999999999998
5 0.109999
6 0.43333299999999997
7 1.7999990000000001
8 5.546666
9 21.939996999999998
10 90.783324
11 1001
};
\addplot[mark=square,mark options={scale=0.8}] table {
1 0.003333
2 0.006666
3 0.026666
4 0.29999899999999996
5 2.3799989999999998
6 8.523332
7 34.293329
8 146.303318
9 1001
};
\addplot [mark=pentagon, mark options={scale=0.9}] table {
1 8.946665
2 1001
};
\addplot [mark=diamond, mark options={scale=0.9}] table {
1 0.016666
2 0.059999
3 3.553332
4 1001
};
\addplot [mark=star, mark options={scale=0.9}] table {
1 0.006666
2 0.126666
3 9.803332
4 664.5966
5 1001
};
\end{semilogyaxis}
\end{tikzpicture}
        \subcaption{$\mathit{early}_{gc}(n,4,2)$ \quad(unsatisfiable)}
    \end{subfigure}
    \caption{Formulas with exponential search space and sub-exponential refutations}
    \label{fig:graphs}
\end{figure*}
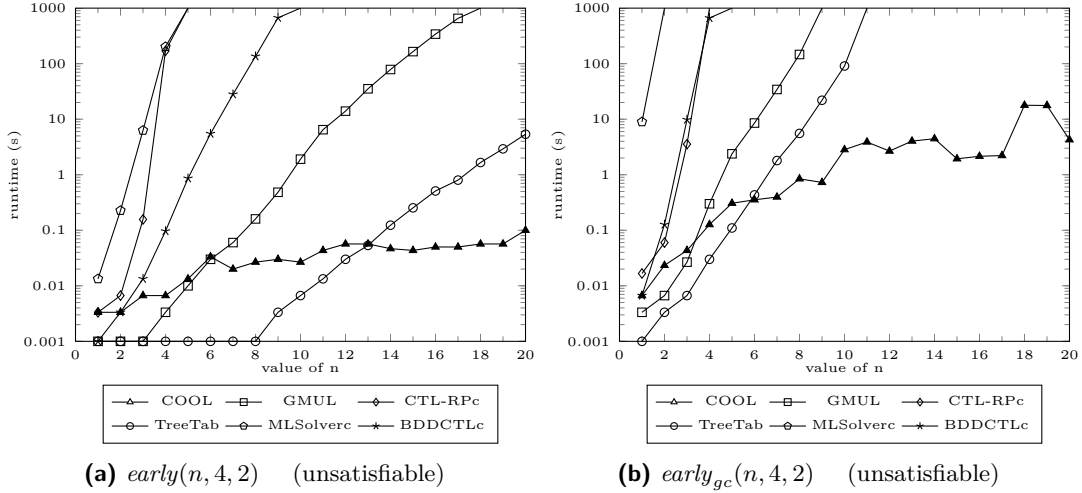

\Ni Figure~\ref{fig:graphs} shows the respective runtimes for these formulas.  
In all cases, COOL finishes before the tableau is fully expanded, while GMUL and
MLSolver will necessarily complete their first pass before being able
to decide the formulas, and hence exhibit exponential behaviour;
TreeTab seems not to benefit substantially from its capability to
close tableaux early.  For the $\mathit{early}_{gc}$ formulas, the
ability to cache previously seen nodes appears to provide COOL with
additional advantages. The $\mathit{early}_{gc}$ series can be converted 
into satisfiable formulas by replacing $AX$ with $EX$, with similar results.

Due to the apparent lack of benchmarking formulas for the
alternation-free $\mu$-calculus and ATL, we compare runtimes on random
formulas for these logics.  For the alternation-free $\mu$-calculus,
formulas were built from 250 random operators (where disjunction and
conjunction are twice as likely as the other operators).  The
experiment was conducted with formulas over three and over ten
propositional atoms, respectively.  MLSolver ran out of memory on
$21\%$ on the formulas over three atoms and on $16\%$ of the formulas
over ten atoms.  COOL answered all queries without exceeding memory
restrictions, and in under one second for all queries but
one. Altogether, COOL was faster than MLSolver for more than $98\%$ of
the random alternation-free formulas, with the median of the ratios of
the runtimes being $0.0431$ in favour of COOL for formulas over three
atoms and $0.0833$ for formulas over ten atoms (recall however that
MLSolver supports the full $\mu$-calculus).  For alternating-time
temporal logic ATL, we compared the runtimes of TATL and COOL on
random formulas consisting of 50 random operators; COOL answered
faster than TATL on all of the formulas, with the median of the ratios
of runtimes being $0.000668$ in favour of COOL.

\section{Conclusion}\label{section:conclusion}
We have presented a tableau-based global caching algorithm of optimal
(\ExpTime) complexity for satisfiability in the alternation-free
coalgebraic $\mu$-calculus; the algorithm instantiates to the
alternation-free fragments of e.g.\ the relational $\mu$-calculus, the
alternating-time $\mu$-calculus (AMC) and the serial monotone
$\mu$-calculus.  Essentially, it simultaneously generates and solves a
deterministic B\"uchi game on-the-fly in a direct construction, in
particular skipping the determinization of co-B\"uchi automata; the
correctness proof, however, is stand-alone. We have generalized the
$2^{\mathcal{O}(n)}$ bound on model size for alternation-free fixpoint
formulas from the relational case to the coalgebraic level of
generality, in particular to the AMC. 

We have implemented the
algorithm as part of the generic solver COOL; the implementation shows
promising
performance for CTL, ATL and the alternation-free relational $\mu$-calculus.
An extension of our global caching algorithm to the full $\mu$-calculus
would have to integrate Safra-style determinization of B\"uchi automata~\cite{Safra88} and
solving of the resulting parity game, both on-the-fly.

\bibliography{coalgml}

\appendix

\newpage

\section{Omitted Proofs and Lemmas}

\subsection{Proofs and Lemmas for Section~\ref{section:prelim}}

\begin{definition}
We let $BV(\psi)$ denote the set of variables $X$ such that $\eta X$
occurs in $\psi$.
\end{definition}

\begin{lemma}[Substitution]
If $BV(\psi)\cap FV(\phi)=\emptyset$, then
\begin{align*}
\sem{\psi}^X_i\sem{\phi}_i = \sem{\psi[X\mapsto \phi]}_i.
\end{align*}
\label{lemma:afsubst}
\end{lemma}
\vspace{-20pt}
\begin{proof}
The proof is by induction over $\psi$. If $\psi = \bot$, $\psi = \top$, $\psi = p$ or $\psi = \neg p$, for $p\in P$,
then $\psi$ is closed so that $\sem{\psi}^X_i\sem{\phi}_i = \sem{\psi} = \sem{\psi[X\mapsto\phi]}_i$. 
If $\psi = X$, then
$\sem{X}^X_i\sem{\phi}_i = \sem{\phi}_i = \sem{X[X\mapsto \phi]}_i$.
If $\psi = Y\neq X$, then
$\sem{Y}^X_i\sem{\phi}_i = \sem{Y}_i = \sem{Y[X\mapsto \phi]}_i$.
The cases for disjunction, conjunction and modal operators are straightforward.
If $\psi = \eta X.\,\psi_1$, then
$\sem{\eta X.\,\psi_1}^X_i\sem{\phi}_i = \sem{\eta X.\,\psi_1}_i = \sem{(\eta X.\,\psi_1)[X\mapsto \phi]}_i$.
If $\psi = \eta Y.\,\psi_1$ for $Y\neq X$, then
$\sem{\eta Y.\,\psi_1}^X_i\sem{\phi}_i =
\eta \sem{\psi_1}^Y_{i[X\mapsto\sem{\phi}_i]}=
\eta \sem{\psi_1[X\mapsto \phi]}^Y_i
= \sem{(\eta Y.\,(\psi_1[X\mapsto \phi]))}_i
= \sem{(\eta Y.\,\psi_1)[X\mapsto \phi]}_i$,
where the second equality holds since for all $A$,
\begin{align*}
\sem{\psi_1}^Y_{i[X\mapsto\sem{\phi}_i]}(A) &=
\sem{\psi_1}_{i[X\mapsto\sem{\phi}_i][Y\mapsto A]} \\
&=\sem{\psi_1}_{i[Y\mapsto A][X\mapsto\sem{\phi}_i]}\\
&=\sem{\psi_1}^X_{i[Y\mapsto A]}\sem{\phi}_i\\
&=\sem{\psi_1}^X_{i[Y\mapsto A]}\sem{\phi}_{i[Y\mapsto A]}\\
&=\sem{\psi_1[X\mapsto \phi]}_{i[Y\mapsto A]}\\
&=\sem{\psi_1[X\mapsto \phi]}^Y_i(A),
\end{align*}
where the second equality holds since $X\neq Y$, the fourth equality
holds since by assumption, $Y\notin FV(\phi)$
and the fifth equality is by the induction hypothesis.
\end{proof}

\Ni We note that by Lemma~\ref{lemma:afsubst},
\begin{align*}
\sem{\eta X.\,\psi}_i = \eta \sem{\psi}^X_i =
\sem{\psi}^X_i\sem{\eta X.\,\psi}_i = \sem{\psi[X\mapsto\eta X.\,\psi]}_i.
\end{align*}

\subsection{Proofs and Lemmas for Section~\ref{section:algorithm}}

\noindent In the following we will consider all deferrals to
be in \emph{decomposed form}, i.e.\
given a formula $\psi$ that belongs to some eventuality $\theta$,
so that $\psi=\alpha\sigma$ for appropriate $\alpha$
and $\sigma$, according to Definition~\ref{defn:affil},
we equivalently represent $\psi$ by the pair $(\alpha,\sigma)$.
This allows us to directly refer to the \emph{base}
$\alpha$ and the \emph{sequence} $\sigma$ of a deferral.
We say that the pair $(\alpha,\sigma)$ \emph{induces} 
the formula $\alpha\sigma$.\\

\Ni\emph{Proof of Lemma~\ref{lemma:afdefcontains}:}
The first part of the Lemma is stated by 
Lemma~\ref{lemma:uniqueevs}.
The proof of the second part is by lexicographic induction
over $(|\sigma|,\alpha)$,
distinguishing cases for $\alpha$. The interesting case is the fixpoint variable case, i.e.\ $\alpha = Y$
for some $Y$. If $|\sigma| = 1$, we have that $\sigma = [Y\mapsto \theta]$ and hence $Y\sigma = \theta$.
If $|\sigma| > 1$, we have $Y\sigma = \chi\kappa$ where $\chi$ is the result of applying the
first substitution from $\sigma$ that touches $Y$ to $Y$ and where $\kappa$ consists of the 
remaining substitutions from $\sigma$. We have $|\kappa|<|\sigma|$
and $(\chi,\kappa)$ is a $\theta$-deferral so that the induction hypothesis finishes the proof.\qed

\begin{lemma}\label{lemma:uniqueevs}
Let $(\alpha,\sigma)$ be an $\theta_1$-deferral and let $(\beta,\kappa)$ be an $\theta_2$-deferral
\sut $\alpha\sigma = \psi = \beta\kappa$. Then $\theta_1=\theta_2$.
\end{lemma}
\begin{proof}
We show that $\theta_2\leq \theta_1$, the other direction is symmetric.
We note that by Lemma~\ref{lemma:afdefcontains}, $\theta_2\leq\psi$. If $\theta_2\leq\alpha$, 
$\theta_2<\theta_1$ and hence $\theta_2\leq \theta_1$, as required.
If $\theta_2\nleq\alpha$, then let $\theta_2=\mu Y.\,\phi$ and $\sigma =[X_1\mapsto\chi_1];\ldots;[X_n\mapsto \chi_n]$
where $\chi_n = \theta_1$.
Since $\theta_2\leq\psi$ but $\theta_2\nleq\alpha$, we are in one of the following two cases:
a) There is a variable $X\in FV(\alpha)$ with $\theta_2\leq X\sigma$ in which case -- since $\theta_2$ 
is irreducible -- $\theta_2\leq \chi_i\leq \theta_1$ for some $1\leq i\leq n$: otherwise there is some 
$\chi_j = \mu Y.\phi_1$ \sut $\mu Y.\phi_1([X_{j+1}\mapsto\chi_{j+1}];\ldots;[X_n\mapsto\chi_n]) = \theta_2$ which is a contradiction
to $\theta_2$ being irreducible;
b) The formula $\alpha$ contains a fixpoint literal $\mu Y.\,\phi_1$ 
with $\phi_1\sigma = \phi$. 
But then $\theta_2 = (\mu Y.\,\phi_1)\sigma$ and $(\mu Y.\,\phi_1,\sigma)$ is a sequence over $\chi_n$
which is a contradiction to $\theta_2$ being irreducible.
\end{proof}

\Ni\emph{Proof of Lemma~\ref{lemma:aftrack}:}
Note that
\begin{align*}
\hat{f}_{X'}(Y') = & (f(X'\cap Y')\cap \overline{F}) 
\cup (f(X')\cap F) \\
\subseteq &  
(f(X\cap Y)\cap \overline{F}) \cup (f(X)\cap F)\\
=&\hat{f}_{X}(Y)
\end{align*}
where the inclusion holds since $X'\cap Y'\subseteq X\cap Y$
and since $f$ is monotone w.r.t. set inclusion so that
$f(X'\cap Y')\subseteq f(X\cap Y)$ and $f(X')\subseteq f(X)$. 
The proof for $\hat{g}$ is analogous.\qed\\

\Ni\emph{Proof of Lemma~\ref{lemma:afsuccunsuccmon}:}
Let $G'\subseteq G$. We show $E_{G'}\subseteq E_{G}$, the proof of
$A_{G'}\subseteq A_{G}$ is analogous.
We denote by $f_C$, and $(\hat{f}_X)_C$ the
respective transitionals with base set $C\subseteq G$ and note
that for all $X,Y\subseteq G$,
\begin{align*}
f_{G'}(Y)\subseteq f_{G}(Y) \quad\text{and}\quad
(\hat{f}_X)_{G'}(Y)\subseteq (\hat{f}_X)_G(Y).
\end{align*}
From this we obtain $\mu((\hat{f}_X)_{G'})\subseteq
\mu((\hat{f}_X)_G)$ by induction; this in turn implies that for all $Y$,
$(X\mapsto\mu((\hat{f}_X)_{G'}))Y\subseteq (X\mapsto\mu((\hat{f}_X)_{G}))Y$.
Induction yields $\nu (X\mapsto\mu((\hat{f}_X)_{G'}))\subseteq
\nu (X\mapsto\mu((\hat{f}_X)_G))$, as required. \qed

\begin{lemma}
Let $G\subseteq\nodes$ be fully expanded
and let $C\subseteq \mathbf{C}_G$ be the base set of $f$ and $g$.
For all sets $Y\subseteq C$,
\begin{center}
$f(Y)=\overline{g(\overline{Y})}$,
\end{center}
where for each $Y'\subseteq C$,
$\overline{Y'}$ denotes the complement of $Y'$ in $C$.
\label{lemma:comp}
\end{lemma}
\begin{proof}
The inclusion ``$\subseteq$'' is immediate.
For the inclusion ``$\supseteq$'', 
let $(\Delta,d)\in\overline{g(\overline Y)}$ so
that it is not the case that there is a $\Sigma\in \concl(\Delta)$
\sut for each $\Gamma\in \Sigma$, $(\Gamma,d_{\Delta\rightsquigarrow \Gamma})\in \overline{Y}$.
Since $G$ is fully expanded, this implies
that for all $\Sigma\in \concl(\Delta)$,
there is a $\Gamma\in \Sigma$ \sut $(\Gamma,d_{\Delta\rightsquigarrow \Gamma}))\in Y$, i.e.\
that $(\Delta,d)\in f(Y)$.
\end{proof}

\begin{lemma}
If $G\subseteq \nodes$ is fully expanded and 
$C\subseteq \mathbf{C}_G$ is the base set of $\hat{f}_X$ 
and $\hat{g}_{\overline{X}}$, then for all sets of nodes $Y\subseteq C$,
\begin{center}
$\hat{f}_X(Y)=\overline{\hat{g}_{\overline{X}}(\overline{Y})}$.
\end{center}
\label{lemma:comptrack}
\end{lemma}
\begin{proof}
Just note that
\begin{align*}
\hat{f}_X(Y) =&(f(X\cap Y)\cap \overline{F})\cup (f(X)\cap F)\\
=&\overline{(g(\overline{X}\cup \overline{Y})\cup F)\cap 
(g(\overline{X})\cup \overline{F})}\\
=&\overline{\hat{g}_{\overline{X}}(\overline{Y})}.
\end{align*}
where the second equality follows, as $G$ 
is fully expanded, from Lemma~\ref{lemma:comp}.
\end{proof}

\Ni\emph{Proof of Lemma~\ref{lemma:afcomplementary}:}
We obtain $E_{G}=\nu(X\mapsto\mu(\hat{f}_X))=\overline{\mu(X\mapsto\nu(\hat{g}_X))}=\overline{A_{G}}$
from Lemma~\ref{lemma:comptrack} which states
that $\hat{f}_X(Y)=\overline{\hat{g}_{\overline{X}}(\overline{Y})}$ 
for all $X\subseteq\mathbf{C}_G$
in combination with the fact that for complementary monotone 
functions $f$ and $g$, $\mu f = \overline{\nu g}$.
\qed\\

\Ni\emph{Proof of Lemma~\ref{lemma:afnonopt}:}
Let $G$ denote the set of nodes which is created by the algorithm
without intermediate propagation -- i.e.\ without step 3) --
and notice that $G$ is fully expanded. 
Let $(\{\phi_0\},d(\{\phi_0\}))\in E_{G}$ and let $G_p$ be the set of nodes
created by any run of the algorithm 
(possibly involving intermediate propagation). We note that $G_p\subseteq G$ 
so that Lemma~\ref{lemma:afsuccunsuccmon} tells us that
$A_{G_p}\subseteq A_{G}$. As $G$ is fully expanded, 
Lemma~\ref{lemma:afcomplementary} states that $A_{G}=\overline{E_{G}}$.
As $(\{\phi_0\},d(\{\phi_0\}))\in E_{G}$, $(\{\phi_0\},d(\{\phi_0\}))\notin A_{G_p}\subseteq
A_{G}=\overline{E_{G}}$, as required.\qed

\subsection{Proofs and Lemmas for Section~\ref{section:soundcomp}}
Throughout this subsection, we fix $N\subseteq\nodes$ to be the fully expanded set of nodes
constructed by a run of the algorithm without intermediate propagation.

\begin{definition}\upshape
Given a substitution $\sigma$, we define the \emph{domain} $\mathit{dom}(\sigma)$ of $\sigma$
as the set of all fixpoint variables that $\sigma$ \emph{touches}, i.e.\
the set of all fixpoint variables $X$ with $\sigma(X)\neq X$.
\end{definition}

Regarding Definition~\ref{defn:sufficiency}, we note that for all $\Gamma\in N$, all
eventualities $\theta$ and all deferrals $\delta$, since 
$d(\Gamma)\cup N(\Gamma,\theta) = \Gamma$, we have
$d(\Gamma)\vdash_\Gamma\delta$ iff $\Gamma\PLentails\delta$.

\begin{lemma}[Syntactic substitution]
If $(\{X\}\cup BV(\psi))\cap\mathit{dom}(\sigma) = \emptyset$ and
for each $Y\in FV(\psi)$, $(\{X\}\cup BV(\psi))\cap FV(\sigma(Y)) = \emptyset$,
\begin{align*}
(\psi\sigma)[X\mapsto(\phi\sigma)] = (\psi[X\mapsto\phi])\sigma.
\end{align*}
\label{fac:afsubstcontract}
\end{lemma}
\begin{proof}
The proof is by induction over $\psi$.
If $\psi=\bot$, $\psi=\top$, $\psi=p$ or $\psi=\neg p$, for $p\in P$,
then $\psi$ is closed and hence
$(\psi\sigma)[X\mapsto(\phi\sigma)] = \psi= (\psi[X\mapsto\phi])\sigma.$
If $\psi = X$, then note that by assumption $X\notin \mathit{dom}(\sigma)$
so that $(X\sigma)[X\mapsto(\phi\sigma)] = X[X\mapsto\phi\sigma] = \phi\sigma =
 (X[X\mapsto\phi])\sigma$. If $\psi = Y \neq X$, then we have by assumption $X\notin FV(\sigma(Y))$
so that $(Y\sigma)[X\mapsto(\phi\sigma)] = \sigma(Y)[X\mapsto\phi\sigma] = \sigma(Y) = Y\sigma = 
 (Y[X\mapsto\phi])\sigma$. The cases for conjunction, disjunction and modal operators are straightforward.
If $\psi = \eta X.\,\psi$, then $((\eta X.\,\psi)\sigma)[X\mapsto(\phi\sigma)] = (\eta X.\,\psi)\sigma =
 ((\eta X.\,\psi)[X\mapsto\phi])\sigma$. If $\psi = \eta Y.\,\psi$ for $X\neq Y$, then we have 
by assumption that $Y\notin \mathit{dom}(\sigma)$ and for any $Z\in FV(\psi)$, $Y\notin FV(\sigma(Z))$
so that
$((\eta Y.\,\psi)\sigma)[X\mapsto(\phi\sigma)] = \eta Y.\,(\psi\sigma)[X\mapsto(\phi\sigma)]
= \eta Y.\,(\psi\sigma[X\mapsto(\phi\sigma)]) = \eta Y.\,((\psi[X\mapsto\phi])\sigma )= (\eta Y.\,(\psi[X\mapsto\phi]))\sigma = ((\eta Y.\,\psi)[X\mapsto\phi])\sigma$, where the third equality is by the induction hypothesis.
\end{proof}

\begin{definition}\upshape
Let $t_1$ and $t_2$ be unfolding trees for $\psi$ and $\phi$. Define
$t_1[X\mapsto t_2]$ as the unfolding tree for $\psi[X\mapsto\phi]$ that is obtained
by replacing every node in $t_1$ that represents a \emph{free} occurrence of $X$ in $\psi$
with $t_2$.
\end{definition}

\begin{lemma}
For each state $x$ and each formula $\psi$ \sut $x\models\psi$, there is a least
unfolding tree $t$ \sut $x\models\psi(t)$.
\end{lemma}
\begin{proof}
We construct $t$ by walking from left to right through all paths in the syntax tree of $\psi$,
assigning numbers to nodes that represent least fixpoint literals. Let $\kappa$ be a position and let $t_\kappa$
denote the tree that has been constructed so far on the walk from the root of the syntax tree
to $\kappa$. We assign $n_\kappa$ to the node at position $\kappa$
if that node represents a least fixpoint literal $\mu X_\kappa.\,\psi_\kappa$ where $n_\kappa$ is the least number
\sut $x\models c_\kappa((\psi_\kappa)_{X_\kappa}^{n_\kappa})$, where $\psi=c(\mu X_\kappa.\,\psi_\kappa)$
and where $c_\kappa$ denotes the context that
is obtained from $c$ by replacing any least fixpoint literal $\mu X_\rho.\,\psi_\rho\leq c$ that already has a number
$n_\rho$ assigned to it in $t_\kappa$ by $(\psi_\rho)_{X_\rho}^{n_\rho}$ and
by replacing any other fixpoint literals in $c$ by their $n$-th unfolding, where $n$ is the size of the finite
model. The unfolding tree that we obtain is by construction the least (w.r.t $<_\psi$) unfolding tree $t$ 
for $\psi$ \sut $x\models\psi(t)$.
\end{proof}

\begin{lemma}
For all $n$, if $X\neq Y$, 
\begin{align*}
(\psi[X\mapsto \phi])^n_Y = \psi^n_Y[X\mapsto \phi].
\end{align*}
\label{fact:unfoldsubst}
\end{lemma}
\begin{proof}
By induction over $n$. If $n=0$, $\bot=\bot$. Otherwise
\begin{align*}
(\psi[X\mapsto \phi])^n_Y & = (\psi[X\mapsto \phi])_Y((\psi[X\mapsto \phi])^{n-1}_Y) \\
& = (\psi[X\mapsto \phi])_Y(\psi^{n-1}_Y[X\mapsto \phi])\\
& = (\psi_Y(\psi^{n-1}_Y))[X\mapsto \phi] = \psi^n_Y[X\mapsto \phi],
\end{align*}
where the second equality is by the induction hypothesis and the third equality
is by Lemma~\ref{fac:afsubstcontract}.
\end{proof}

\begin{lemma}Let
$t_1$ be an unfolding tree for $\psi$ and
let $t_2$ be an unfolding tree for $\phi$. Then
\begin{align*}
(\psi[X\mapsto \phi])(t_1[X\mapsto t_2]) = (\psi(t_1))[X\mapsto \phi(t_2)].
\end{align*}
\end{lemma}
\begin{proof}
The proof is by standard induction over $\psi$. We consider the
only interesting case, i.e.\ the case that $\psi = \mu Y.\psi_1$ where
$X\neq Y$. Then 
\begin{align*}
(\mu Y.\psi_1[X\mapsto \phi])(t_1[X\mapsto t_2]) & = (\mu Y.(\psi_1[X\mapsto \phi]))(t_1[X\mapsto t_2]) \\
&= ((\psi_1[X\mapsto \phi])(t_3[X\mapsto t_2]))^n_Y \\
&= ((\psi_1(t_3))[X\mapsto \phi(t_2)])^n_Y\\
&= ((\psi_1(t_3)))^n_Y[X\mapsto \phi(t_2)]\\
&= (\mu Y.\psi_1(t_1))[X\mapsto \phi(t_2)]
\end{align*}
where
$t_3$ is the child of the root of $t_1$.
The third equality is by the induction hypothesis and the fourth
equality is by Lemma~\ref{fact:unfoldsubst}.
\end{proof}

\begin{lemma}
Let $t$ and $s$ be unfolding trees for $\phi_1=\eta X.\,\psi\sigma$ and
$\phi_2=\psi(\eta X.\,\psi,\sigma)$, respectively.
Furthermore, let $t(\epsilon,\phi_1) =n+1$ and $s(\tau,\phi_1) = n$ for all
positions $\tau$ at which $\phi_1$ occurs in $\phi_2$; also let 
$t(\kappa,\chi) = s(\tau,\chi)$ 
for all least fixpoint literals $\chi$ occurring in $\phi_1$ at some position $\kappa\neq\epsilon$
and all $\tau$ \sut $\chi$ occurs in $\phi_2$
at position $\tau$ and either $\kappa = 0\tau$ or $\tau = \rho\kappa$ where $X$ occurs freely in
$\psi$ at position $\rho$. Then
\begin{align*}
x\models \eta X.\,\psi\sigma(t) \text{ implies }
x\models (\psi(\eta X.\,\psi,\sigma))(s).
\end{align*}
\label{lemm:unfoldings}
\end{lemma}

\begin{proof}
So let $t(\epsilon,\eta X.\,\psi\sigma) = n + 1 = s(\tau,\eta X.\,\psi\sigma) + 1$
for all appropriate $\tau$.
Let $t_1$ denote the child of the root of $t$ and let $s_1$,
$s_2$ and $s_3$ denote subtrees of $s$ \sut $s = s_1[X\mapsto s_2]$
and $s_3$ is the child of the root of $s_2$.
Then 
\begin{align*}
\eta X.\,\psi\sigma(t) & = (\psi\sigma(t_1))_X^{n+1} \\
& = (\psi\sigma(t_1))_X((\psi\sigma(t_1))_X^{n})
\end{align*}
and
\begin{align*}
(\psi(\eta X.\,\psi,\sigma))(s) & =
((\psi[X\mapsto\eta X.\,\psi])\sigma))(s) \\
& = (\psi\sigma[X\mapsto\eta X.\,\psi\sigma])(s)\\
& = (\psi\sigma(s_1))([X\mapsto\eta X.\,\psi\sigma](s_2)) \\
& = (\psi\sigma(s_1))_X(\eta X.\,\psi\sigma(s_2)) \\
& = (\psi\sigma(s_1))_X((\psi\sigma(s_3))_X^{n}),
\end{align*}
where the fifth equality holds since $s_2(\epsilon,\eta X.\,\psi\sigma)=n$.
As $\psi\sigma$ does not contain $\eta X.\,\psi\sigma$ 
and $s$ and $t$ agree on all other fixpoint literals,
$t_1 = s_1 = s_3$, which finishes the proof.
\end{proof}

\begin{definition}[Realization]\label{defn:afrealized}\upshape
  The set of \emph{$\mathcal K$-realized} nodes is
  \begin{equation*}
    M=\{(\Gamma,d)\mid \Gamma\in N, d\subseteq d(\Gamma),\exists x\in W.\, \forall \phi.\Gamma\PLentails\phi\Rightarrow x\models_W \phi\}.
  \end{equation*}
\end{definition}

\begin{definition}[Rank]\upshape
The \emph{rank} $\rank(\psi)$ of a formula $\psi$ is the 
depth of nesting of modal operators in it.
Given a set $d$ of deferrals and a state $x\in W$ \sut
$x\models\alpha\sigma$ for each $(\alpha,\sigma)\in d$, we put
\begin{equation*}
\rank(d,x)=\max\{\rank(\mathit{unf}((\alpha,\sigma),x))\mid
(\alpha,\sigma)\in d\}.
\end{equation*}
For $(\Gamma,d')\in M$, we put
\begin{equation*}
\rank(d,\Gamma)=\min\{\rank(d,x)\mid \forall \phi.\Gamma\PLentails\phi \Rightarrow x\models \phi\}.
\end{equation*}
\end{definition}

\begin{corollary}
Let $x\models(\eta X.\,\psi)\sigma$. Then
\begin{align*}
\rank(\mathit{unf}((X,(\eta X.\,\psi,\sigma),x)))
& \geq \rank(\mathit{unf}((\psi,(\eta X.\,\psi,\sigma),x))).
\end{align*}
\label{cor:unfold}
\end{corollary}
\begin{proof}
Let $t$ and $s$ be the least unfolding trees for 
$X(\eta X.\,\psi,\sigma)=\eta X.\,\psi\sigma$ and 
$\psi(\eta X.\,\psi,\sigma)$ \sut $x\models \eta X.\,\psi\sigma(t)$
and $x\models (\psi(\eta X.\,\psi,\sigma))(s)$, respectively.
Lemma~\ref{lemm:unfoldings} finishes the proof as it states that $s$ can be
chosen to agree with $t$ on all least fixpoint literals except for $\eta X.\,\psi\sigma$
for which we have $t(\epsilon,\eta X.\,\psi\sigma) = s(\kappa,\eta X.\,\psi\sigma) + 1$
for any suitable $\kappa$;
thus $(\psi(\eta X.\,\psi,\sigma))(s)$ has a rank that is not
greater than the rank of $\eta X.\,\psi\sigma(t)$, as required.
\end{proof}

\begin{lemma} For all deferrals $(\alpha,\sigma)$ and all unfolding trees
$t_{\alpha\sigma}$,
\begin{align*}
\sem{\alpha\sigma(t_{\alpha\sigma})}\subseteq\sem{\alpha\sigma}.
\end{align*}
\label{fac:afunfold}
\end{lemma}
\begin{proof}
This lemma follows by induction over $\alpha\sigma$ from $\sem{\psi_{X}^{n}}\subseteq \sem{\mu X.\,\psi}$.
\end{proof}




\begin{definition}[(Pseudo-)Theory]\upshape
We define the \emph{pseudo-theory} $\Gamma\PLentails$ of a node 
$\Gamma\in N$ as
\begin{equation*}
\Gamma\PLentails\;\;\;=\;\{\phi\in\FLphi\mid \Gamma\PLentails\phi\},
\end{equation*}
and the \emph{theory} $x\models$ of a state $x\in W$ as
\begin{equation*}
x\models\;\;\;=\;\{\phi\in\FLphi\mid x\models\phi\}.
\end{equation*}
Given a node $\Gamma\in N$ and a state $x\in W$, we write $\Gamma\subseteq x$ if
$(\Gamma\PLentails)\subseteq (x\models)$, equivalently
$\Gamma\subseteq(x\models)$.
\end{definition}
\noindent
Recall that $M$ denotes the set of $\mathcal K$-realized nodes 
(cf. Definition~\ref{defn:afrealized}) and note that 
\begin{equation*}
M=\{(\Gamma,d)\mid \Gamma\in N, d\subseteq d(\Gamma),\exists x\in W.\,\Gamma\subseteq x\}.
\end{equation*}

\begin{lemma} Let $x\in W$, $(\Delta,d)\in M\cap \snodes\times\snodes$ and $\Delta\subseteq x$.
  Given a set $B_{\langle a \rangle\alpha}\subseteq W$ for each $\langle a\rangle \alpha\in \Delta$,
	a set $B_{[a]\alpha}\subseteq W$ for each $[a]\alpha\in \Delta$ \sut
	\begin{align*}
  \langle a \rangle\alpha\in \Delta\Rightarrow \exists y\in R_a(x).y\in B_{\langle a\rangle\alpha}\\
	[a]\alpha\in \Delta\Rightarrow \forall y\in R_a(x).y\in B_{[ a]\alpha},
	\end{align*} and a modal rule 
	\begin{align*}
	(\Gamma,[a]\psi_1,\ldots,[a]\psi_n,\langle a\rangle\psi/
	\psi_1,\ldots,\psi_n,\psi)
	\end{align*}
	with $\Gamma,[a]\psi_1,\ldots,[a]\psi_n,\langle a\rangle\psi = \Delta$,
	we have $\{\psi_1,\ldots,\psi_n,\psi\}=\Theta\in N$ and there is a state $z\in W$ \sut
	$\Theta\subseteq z$ and $z\in\bigcap_{1\leq i\leq n} B_{[a]\psi_i}\cap B_{\langle a\rangle\psi}$.
  
\label{lemm:afstateformstate}
\end{lemma}

\begin{proof} 
  As $N$ is fully expanded, $\{\psi_1,\ldots,\psi_n,\psi\}=\Theta\in N$.
	As $\langle a\rangle \psi\in \Delta$, there is by assumption a state $z\in B_{\langle a\rangle\psi}$.
	Since $[a]\psi_i\in \Delta$ for $1\leq i\leq n$, we have by assumption that $z$ is also contained in
	$\bigcap_{1\leq i\leq n} B_{[a]\psi_i}$, as required.
	
\end{proof}

\begin{definition}\upshape
 We denote by $\mathit{u_f}(\phi)$ and $\mathit{u_p}(\phi)$ the numbers of     
  unguarded occurrences of fixpoint and propositional operators
  in $\phi$, respectively.
\end{definition}

\Ni\emph{Proof of Theorem~\ref{thm:afsatsucc}:}
  It suffices to show that $\mathcal K$-realized nodes are successful,
  i.e.\ $M\subseteq E_S=\nu(X\mapsto \mu(\hat{f}_X))$. We use coinduction, i.e.\ show
  that $M$ is a postfixpoint of $(X\mapsto \mu(\hat{f}_X))$, i.e.\
  $(\Delta,d)\in\mu(\hat{f}_M)$ for all $(\Delta,d(\Delta))\in M$. 
  We show the more general property that for all $\Delta\in N$ and all $d\subseteq d(\Delta)$, $(\Delta,d)\in\mu(\hat{f}_M)$
	and proceed by induction over the triple
  $(\rank(d,\Delta),\mathit{u_f}(\Delta),\mathit{u_p}(\Delta))$ in lexicographic order $<_l$.
	If $d=\emptyset$, then $(\Delta,d)\in\hat{f}_M(\mu(\hat{f}_M))$ if $(\Delta,d)\in f(M)$ which
	is implied by Lemma~\ref{lemm:afea}. If $d\neq\emptyset$,
  $\rank(d,\Delta)>0$.  We distinguish two cases:
\begin{itemize}
\item If $\Delta$ is a not state node, then let $y$ be a state with $\Delta\subseteq y$.
We note that $\mathit{u_f}(\Delta)>0$ or $\mathit{u_p}(\Delta)>0$.
Let $\Delta=\{\phi_1,\ldots,\phi_o\}$.
In order to show that $(\Delta,d)\in\hat{f}_M(\mu(\hat{f}_M))$,
we consider any non-modal rule that matches $\Delta$
and show that it has a conclusion $\Theta$ \sut $(\Theta,d_{\Delta\rightsquigarrow \Theta})\in \mu(\hat{f}_M)$.
To this end we distinguish upon the rule that is being applied.
\begin{itemize}
\item $(\bot)$, $(p,\neg p)$: Thes rules are not applicable to $\Delta$ since $\Delta\subseteq y$
and $y\not\models \bot$ as well as $y\not\models p\wedge\neg p$ for any $p$.
\item $(\wedge)$: Then there is a formula $\phi_i=\psi_1\wedge\psi_2\in \Delta$ 
and the rule leads -- since $N$ is fully expanded -- to the node 
$\Theta\in N$ with 
\begin{align*}
\Theta=\{\phi_1,\ldots,\phi_{i-1},\psi_1,\psi_2,\phi_{i+1}\ldots,\phi_o\}.
\end{align*}
We note that $\mathit{u_f}(\Theta)=\mathit{u_f}(\Delta)$,
$\mathit{u_p}(\Theta)<\mathit{u_p}(\Delta)$ and 
$\Theta\subseteq y$, i.e.\ $(\Theta,d_{\Delta\rightsquigarrow \Theta})\in M$; also $\rank(d_{\Delta\rightsquigarrow \Theta},\Theta)\leq
\rank(d,\Delta)$. By the induction hypothesis,
$(\Theta,d_{\Delta\rightsquigarrow \Theta})\in\mu(\hat{f}_M)$, as required. 
\item $(\vee)$: Then there is a formula $\phi_i=\psi_1\vee\psi_2\in s$ 
and the rule leads -- since $N$ is fully expanded -- to the two nodes 
$\Theta_1,\Theta_2\in N$ with
\begin{align*}
\Theta_1&=\{\phi_1,\ldots,\phi_{i-1},\psi_1,\phi_{i+1}\ldots,\phi_o\} \quad \text{and}\\
\Theta_2&=\{\phi_1,\ldots,\phi_{i-1},\psi_2,\phi_{i+1}\ldots,\phi_o\}.
\end{align*}
We note that $\mathit{u_f}(\Theta_1)=\mathit{u_f}(\Theta_1)=\mathit{u_f}(\Delta)$, 
$\mathit{u_p}(\Theta_1)< \mathit{u_p}(\Delta)$ and $\mathit{u_p}(\Theta_2)<\mathit{u_p}(\Delta)$; also 
$\Theta_1\subseteq y\models$ or $\Theta_2\subseteq y\models$
so that there is an $i\in\{1,2\}$ with 
$\Theta_i\subseteq y$, i.e.\ with $(\Theta_i,d_{\Delta\rightsquigarrow \Theta_i})\in M$; furthermore, $\rank(d_{\Delta\rightsquigarrow \Theta_i},\Theta_i)\leq
\rank(d,\Delta)$. By the induction hypothesis, $(\Theta_i,d_{\Delta\rightsquigarrow \Theta_i})\in\mu(\hat{f}_M)$, as required. 


\item $(\eta)$: Then there is a formula $\phi_i=\eta X.\psi\in \Delta$ 
and the rule leads -- since $N$ is fully expanded -- to the node 
$\Theta\in N$ with
\begin{align*}
\Theta=\{\phi_1,\ldots,\phi_{i-1},\psi[X\mapsto\eta X.\psi],\phi_{i+1}\ldots,\phi_o\}.
\end{align*}
We note that $\mathit{u_f}(\Theta)< \mathit{u_f}(\Delta)$ and 
$\Theta\subseteq y$ so that $(\Theta,d_{\Delta\rightsquigarrow \Theta})\in M$.
Let $\chi$ abbreviate $\eta X.\psi$; if $\eta = \nu$, 
$\chi$ is not induced by any deferral from $d$ so that
$\rank(d_{\Delta\rightsquigarrow \Theta},\Theta) = \rank(d,\Delta)$. If $\eta = \mu$,
then we show that $\rank(d_{\Delta\rightsquigarrow \Theta},\Theta) 
\leq \rank(d,\Delta)$.
Notice that we can choose a sequence 
$\sigma = [X_1\mapsto \chi_1];\ldots;[X_n\mapsto \chi_n]$
that sequentially unfolds some
eventuality $\chi_n$ and a formula $\psi_1$ \sut $\mu X.\,\psi_1<_f \chi_1$ and $\psi_1\sigma=\psi$; then
$(X,[X\mapsto \mu X.\,\psi_1];\sigma)$ is a deferral that induces $\chi = \mu X.\,\psi_1\sigma$ and
$(\psi_1,[X\mapsto \mu X.\,\psi_1];\sigma)$ is a deferral that induces $(\psi_1[X\mapsto\mu X.\,\psi_1])\sigma = 
\psi[X\mapsto \mu X.\,\psi]$ so that if $(X,[X\mapsto \mu X.\,\psi_1];\sigma)\in d$,
$(\psi_1,[X\mapsto \mu X.\,\psi_1];\sigma)\in d_{\Delta\rightsquigarrow \Theta}$.
By Corollary~\ref{cor:unfold}, $\rank(\mathit{unf}((X,[X\mapsto \mu X.\,\psi_1];\sigma),y)) \geq
\rank(\mathit{unf}((\psi_1,[X\mapsto \mu X.\,\psi_1];\sigma),y))$ which implies --
since $(X,[X\mapsto \mu X.\,\psi_1];\sigma)$ is the only deferral that changed from $\Delta$ to $\Theta$ -- that we have
$\rank(d_{\Delta\rightsquigarrow \Theta},\Theta)\leq\rank(d,\Delta)$. The induction hypothesis implies
$(\Theta,d_{\Delta\rightsquigarrow \Theta})\in\mu(\hat{f}_M)$, as required. 

\end{itemize}

\item If $\Delta$ is a state node, then let $x$ be a state with $\Delta\subseteq x$ and
  $\rank(d,\Delta)=\rank(d,x)$. In order to show that $(\Delta,d)\in \hat{f}_M(\mu(\hat{f}_M))$,
  we show that for all modal rules that match $\Delta$, there is a conclusion $\Theta$ of the rule application with
  $(\Theta,d_{\Delta\rightsquigarrow \Theta})\in\mu(\hat{f}_M)$.
  Consider any $(\langle a\rangle)$-rule
	\begin{align*}
	(\Gamma,[a]\psi_1,\ldots,[a]\psi_n,\langle a\rangle\psi\,/\,
	\psi_1,\ldots,\psi_n,\psi) 
	\end{align*}
	with
	$\Delta=\Gamma,[a]\psi_1,\ldots,[a]\psi_n,\langle a\rangle\psi$.
	We define for each $(\langle a\rangle\beta,\sigma)\in d$ the set
  $B_{\langle a\rangle\beta\sigma}= \sem{\beta\sigma(t)}$
	where $\mathit{unf}((\langle a\rangle\beta,\sigma),x) = \langle a\rangle\beta\sigma(t)$.
	We also define for each
  $([a]\beta,\sigma)\in d$ the set
  $B_{[a]\beta\sigma}= \sem{\beta\sigma(t)}$
	where $\mathit{unf}(([a]\beta,\sigma),x) = [a]\beta\sigma(t)$.
  By Fact~\ref{fac:afunfold},
  $\sem{\beta\sigma(t)} \subseteq\sem{\beta\sigma}$. For
  each $\langle a\rangle\beta\in \Delta$ that is not
  induced by a deferral from $d$, we define
  $B_{\langle a\rangle\beta}=\sem{\beta}$, and analogously we
	put $B_{[a]\beta}=\sem{\beta}$ for each $[a]\beta\in \Delta$ that
	is not induced by a deferral from $d$. Note how for each
  $\langle a \rangle\beta\in \Delta$, there is an $y\in R_a(x)$ with
  $y\in B_{\langle a\rangle\beta}$: If
  $\langle a\rangle\beta\in \Delta$ is not induced by a deferral, note
  that $\Delta\subseteq x$ so that
  $x\in\sem{\langle a\rangle\beta}$. Otherwise,
  note that $B_{\langle a\rangle\beta\sigma} = \sem{\beta\sigma(t)}$
  where
  $x\in\sem{\langle a\rangle\beta\sigma(t)}$ which is
  the case iff there is a $y\in R_a(x)$ with $y\in\sem{\beta\sigma(t)}=
  B_{\langle a\rangle\beta\sigma}$, as required.
	For each $[a]\beta\in s$, one shows analogously that for all $y\in R_a(x)$,
  $y\in B_{[a]\beta}$.
  Thus by Lemma~\ref{lemm:afstateformstate}, $\{\psi_1,\ldots,\psi_n,\psi\}=\Theta\in M$ and there is
  a state $z\in W$ with $\Theta\subseteq z$ \sut $\bigcap_{1 \leq i\leq n}B_{[a]\psi_i}\cap B_{\langle a\rangle\psi}$.
  The induction hypothesis implies
  $(\Theta,d_{\Delta\rightsquigarrow \Theta})\in\mu(\hat{f}_M)$ if $\rank(d_{\Delta\rightsquigarrow \Theta},\Theta)<\rank(d,\Delta)$.
  We convince ourselves that indeed
  $\rank(d_{\Delta\rightsquigarrow \Theta},\Theta)\leq \rank(d_{\Delta\rightsquigarrow \Theta},y) < \rank(d,x) =
  \rank(d,\Delta)$: Recall that
$\rank(d_{\Delta\rightsquigarrow \Theta},y)=
\mathit{max}\{\rank(\mathit{unf}((\alpha,\sigma),y))\mid
(\alpha,\sigma)\in d_{\Delta\rightsquigarrow \Theta}\}$.
Take any $(\alpha,\sigma)\in d_{\Delta\rightsquigarrow \Theta}$ for which
$\rank(\mathit{unf}((\alpha,\sigma),y)) =
\rank(d_{\Delta\rightsquigarrow \Theta},y)$
and consider $(\langle a\rangle\alpha,\sigma)\in d$ (the case for $([a]\alpha,\sigma)\in d$
is analogous, using the upcoming argumentation);
if no such deferral exists, $d_{\Delta\rightsquigarrow \Theta}=\emptyset$
and Lemma~\ref{lemm:afea} finishes the proof. Otherwise
let $p=\rank(\mathit{unf}((\langle a\rangle\alpha,\sigma),x))$ 
and let $q=\rank(\mathit{unf}((\alpha,\sigma),y))$.
Recall that $y\in B_{\langle a\rangle \alpha\sigma} = \sem{\alpha\sigma(t)}$ 
so that $\rank(\mathit{unf}((\alpha,\sigma),y))\leq
\rank(\alpha\sigma(t))$ and hence $q< p$. Thus $\rank(\mathit{unf}((\alpha\sigma),y))<
\rank(\mathit{unf}((\langle a\rangle\alpha\sigma),x))$.
Hence
\begin{align*}
\rank(d_{\Delta\rightsquigarrow \Theta},y)&
=\rank(\mathit{unf}((\alpha,\sigma),y))\\
&<\rank(\mathit{unf}(({\langle a\rangle}\alpha,\sigma),x))\\
&\leq\rank(d,x),
\end{align*}
as required.
\end{itemize}
This finishes the proof.
\qed

\begin{lemma} For each focused node $(\Delta,d)\in M$ and each
$\Sigma\in\concl(\Delta)$, there is a $\Theta\in \Sigma$
\sut $(\Theta,d_{\Delta\rightsquigarrow \Theta})\in M$.
\label{lemm:afea}
\end{lemma}
\begin{proof}
Let $(\Delta,d)\in M$ and $\Sigma\in\concl(\Delta)$.
If $\Delta$ is a state node, $\Sigma$ contains just the conclusion $\Theta$
of a modal rule $(\Gamma,[a]\psi_1,\ldots,[a]\psi_n,\langle a\rangle\psi/\psi_1,\ldots,\psi_n,\psi:=\Theta)$
with $\Delta=\Gamma,[a]\psi_1,\ldots,[a]\psi_n,\langle a\rangle\psi$. Since $N$ is fully expanded, $\Theta\in N$.
As $(\Delta,d)\in M$, there is a state $x$ \sut $x\models \langle a\rangle\psi$, i.e.\
there is a state $y\in R_a(x)$ \sut $y\models \psi$. As $x\models [a]\psi_i$, $y\models \psi_i$, for $1\leq i \leq n$,
so that $\Theta\subseteq x$, showing $(\Theta,d_{\Delta\rightsquigarrow \Theta})\in M$, as required.
If $\Delta$ is not a state node, just note that for all $y$,
$y\models$ is closed under propositional breakdown and unfolding of fixpoint literals.
\end{proof}


\begin{definition}\upshape
 A finite set of formulas $\Psi$ \emph{propositionally entails} a
  finite set $\Phi$ of formulas (written $\Psi\PLentails\Phi$) if
  $\Psi\PLentails\bigwedge\Phi$. 

\end{definition}
	
\Ni\emph{Proof of Lemma~\ref{lemm:aftableau-existence}:}
Recall that $E=E_G$.
First note that $|D|\leq |E|\leq 3^{|\phi_0|}$. 
We proceed in two steps: in the first step, we construct a relation $L\subseteq D\times D$;
in the second step, we show that $L$ is a timed-out tableau.
\begin{enumerate}
\item For any $(\Delta,d)\in D$, $(\Delta,d)\in E = 
\nu (X\mapsto\mu(\hat{f}_X))=(X\mapsto\mu(\hat{f}_X))(E)=\mu (\hat{f}_E)=
(\hat{f}_E)^{n}(\emptyset)$ for some $n$.
Let $\concl(\Delta)=\{\Sigma_1,\ldots,\Sigma_j\}$.
If $n = 0$, $(\Delta,d)\notin (\hat{f}_E)^0(\emptyset)=\emptyset$ so that
there is nothing to show.
If $n > 0$, $(\Delta,d)\in\hat{f}_E((\hat{f}_E)^{n-1}(\emptyset))$.
If $d=\emptyset$, then $(\Delta,d)\in f(E)\cap F$,
i.e.\ there is, for each $i$, a $\Gamma\in\Sigma_i$ \sut
$(\Gamma,d_{\Delta\rightsquigarrow \Gamma})\in E$. Notice that since $d=\emptyset$, 
$d_{\Delta\rightsquigarrow \Gamma}=d(\Gamma)$. As $(\Delta,d)\in(\hat{f}_E)^n(\emptyset)$,
this implies by Lemma~\ref{lemm:afdecrstateexists} 
that there is a state node $\Theta_i$ with $\Theta_i\PLentails \Gamma$. Notice that
$d(\Theta_i)\vdash_{\Theta_i} d(\Gamma)$. Put
$L(\Delta,d)=\{(\Theta_1,d(\Theta_1)),\ldots,(\Theta_j,d(\Theta_j))\}$.
If $d\neq \emptyset$, $(\Delta,d)\in f((\hat{f}_E)^{n-1}(\emptyset))$,
i.e.\ there is, for each $i$, a $\Gamma\in\Sigma_i$ \sut
$(\Gamma,d_{\Delta\rightsquigarrow \Gamma})\in (\hat{f}_E)^{n-1}(\emptyset)$.
If $n-1=0$, $\concl(\Delta)=\emptyset$ and we put $L(\Delta,d)=\emptyset$. Otherwise
Lemma~\ref{lemm:afdecrstateexists} implies that there is a state
node $\Theta_i$ with $\Theta_i\PLentails \Gamma$ and a set
$d_i\subseteq d(\Theta_i)$ with $d_i\vdash_{\Theta_i}d_{\Delta\rightsquigarrow \Gamma}$; for step 2), we
note that the Lemma also tells us that $(\Theta_i,d_i)\in(\hat{f}_E)^{n-1}(\emptyset)$.
Put $L(\Delta,d)=\{(\Theta_1,d_1),\ldots,(\Theta_j,d_j)\}$.

\item  We show that $L$ is a timed-out tableau by proving the stronger property
that for all $(\Delta,d)\in D$ and all $d'\subseteq d(\Delta)$, there
is some $m$ \sut $(\Delta,d)\in \mathit{to}(d',m)$.
To this end we distinguish two cases. In case a), $d=d'$, while
in case b), $d\neq d'$.
In both cases, $(\Delta,d)\in E = 
\nu (X\mapsto\mu(\hat{f}_X))=(X\mapsto\mu(\hat{f}_X))(E)=\mu (\hat{f}_E)=
(\hat{f}_E)^{n}(\emptyset)$ for some $n$.
If $d'=\emptyset$, $(\Delta,d)\in \mathit{to}(\emptyset,m)=D$ for any $m$ and we are done.
If $d'\neq \emptyset$, then we proceed by induction over $n$. 
Let $L(\Delta,d)=\{(\Theta_1,d_1),\ldots,(\Theta_j,d_j)\}$.
If $n=0$, $\concl(\Delta)=L(\Delta,d)=\emptyset$ in which case there is nothing to show,
or $(\Delta,d)\in f(E)\cap F$, so that $d=\emptyset$.
Considering the latter situation,
if we are in case a), $d'=\emptyset$ and $(\Delta,\emptyset)
\in \mathit{to}(\emptyset,m)=D$ for any $m$ so that we are done.
If we are in case b), recall from step 1) that $d_1=d(\Theta_1),\ldots,d_j=d(\Theta_j)$;
we proceed as in case a), having to show that for all $1\leq i\leq j$,
$(\Theta_i,d_i)\in\mathit{to}(d_i,m)$ for some $m$.
If $n>0$, recall from step 1) that $L(\Delta,d)=\{(\Theta_1,d_1),\ldots,(\Theta_j,d_j)\}$,
where $(\Theta_i,d_i)\in(\hat{f}_E)^{n-1}(\emptyset)$.
By the induction hypothesis, $(\Theta_i,d_i)\in\mathit{to}(d(\Theta_i),m)$ for some $m$,
as required.
\end{enumerate}
Thus we have constructed a relation $L$ over $D$ -- where
$D$ has size at most $3^{|\phi_0|}$ -- and shown it to be a timed-out tableau.
\qed

\begin{lemma}
Given a set $X\subseteq \mathbf{C}_G$ and
a focused node $(\Delta,d)\in (\hat{f}_X)^n(\emptyset)$, there
is a state node $\Theta$ and a set of deferrals $d'\subseteq d(\Theta)$ \sut
$\Theta\PLentails \Delta$, $d'\vdash_\Theta d$ and $(\Theta,d')\in (\hat{f}_X)^n(\emptyset)$. 
\label{lemm:afdecrstateexists}
\end{lemma}
\begin{proof}
  We proceed by induction over the pair $(\mathit{u_f}(\Delta),\mathit{u_p}(\Delta))$ in
  lexicographic order $<_l$.
  If $\mathit{u_f}(\Delta)=0$ and $\mathit{u_p}(\Delta)=0$, then $\Delta$ is a state node so
  that it suffices to put $\Theta=\Delta$ and $d'=d$. Otherwise $\Delta$ is not a state node
	so that at least one rule matches $\Delta$.
  Let $\Sigma\in\concl(\Delta)\neq\emptyset$. Since $\Delta\in(\hat{f}_X)^n(\emptyset)$, there is a
  $\Gamma\in \Sigma$ with
  $(\Gamma,d_{\Delta\rightsquigarrow \Gamma})\in X\cap
  (\hat{f}_X)^{n-1}(\emptyset)\subseteq(\hat{f}_X)^{n}(\emptyset)$.
  Also $d_{\Delta\rightsquigarrow \Gamma}\subseteq d(\Gamma)$ and since $\Gamma$ is
  obtained from $\Delta$ as conclusion of a non-modal rule,
  $\Gamma\PLentails \Delta$. We note that since $\Gamma\PLentails \Delta$ and
  $d\subseteq d(\Delta)\subseteq \Delta$, we have
  $d_{\Delta\rightsquigarrow \Gamma}\vdash_\Gamma d$. As the non-modal rule
  either unfolds one unguarded fixpoint literal which then becomes guarded
  or removes one unguarded propositional connective from $\Delta$, we have
  that $(\mathit{u_f}(\Gamma),\mathit{u_p}(\Gamma))<_l
  (\mathit{u_f}(\Delta),\mathit{u_p}(\Delta))$ so that by induction we have a state 
  node $\Theta$ and a set $d'\subseteq d(\Theta)$ with $\Theta\PLentails \Gamma$,
  $d'\vdash_\Theta d_{\Delta\rightsquigarrow \Gamma}$ and $(\Theta,d')\in(\hat{f}_X)^n(\emptyset)$. By
  transitivity of propositional entailment, $\Theta\PLentails \Delta$ and
  $d'\vdash_\Theta d$ so that we are done.
\end{proof}


\begin{definition}\upshape
  \Ni A formula $\phi$ is \emph{(closed-)respected} if
  $\psem{\eta X.\,\psi}\subseteq\sem{\eta X.\,\psi}$ for each (closed)
  fixpoint literal $\eta X.\,\psi\leq\phi$. We extend the notion of pseudo-extension
	to sets $\Psi$ of formulas by putting
  \begin{math}\textstyle
    \psem{\Psi}=\bigcap_{\psi\in\Psi}\psem{\psi}.
  \end{math}

\end{definition}

\begin{definition}\label{defn:sequencesem}
Given a sequence $\sigma$, we define the interpretation $\widehat{\sigma}$ 
as $\widehat{\sigma}(Y) = \psem{\sigma(Y)}$, for each $Y\in\mathfrak{V}$.
We put $\sem{\alpha}{\widehat{\sigma}} = \sem{\alpha}_{\widehat{{\sigma}}}$.
\end{definition}

\begin{lemma}
Let $\psi$ be a closed-respected formula. Then
\begin{description}
\item[a)] $\psem{\nu X.\,\psi}\subseteq 
\sem{\nu X.\,\psi}$ and
\item[b)] $\psem{\mu X.\,\psi} \subseteq 
\sem{\mu X.\,\psi}$.
\end{description}
\label{lemm:affps}
\end{lemma}	

\begin{proof}
For a), we note that $\sem{\nu X.\,\psi} = \nu \sem{\psi}_X$.
Hence we proceed by coinduction, i.e.\ we show that $\psem{\nu X.\,\psi}\subseteq
\sem{\psi}_X\psem{\nu X.\,\psi} = \sem{\psi}\widehat{(\nu X.\,\psi)}$.
We have $\psem{\nu X.\,\psi} = \psem{\psi[X\mapsto \nu X.\,\psi]} = 
\psem{\psi(\nu X.\,\psi)}$. As $\psi<_f\nu X.\psi$,
Lemma~\ref{lemma:afgfpsat} finishes the case.
For b), notice that 
\begin{align*}
\psem{\mu X.\,\psi} = \psem{\psi[X\mapsto\mu X.\,\psi]}
= \psem{\psi(\mu X.\,\psi)}
\end{align*}
and that $(\psi,(\mu X.\,\psi))$ is
$\mu X.\,\psi$-deferral. Also $\sem{\mu X.\,\psi} = 
\sem{\psi(\mu X.\,\psi)}$.
Let $v \in \psem{\psi(\mu X.\,\psi)}$ and note that by definition
of sufficiency (Definition~\ref{defn:sufficiency}), $d(l(v))\vdash_{l(v)} \psi(\mu X.\,\psi)$.
Since $v \in W$ and since $L$ is a timed-out tableau, we have $v \in to(d(\Delta),n)$
for some $n$. By Lemma~\ref{lemma:aftableausat}, $v \in \sem{\psi(\mu X.\,\psi)}$,
as required.
\end{proof}

\begin{lemma}
For all $\sigma = [X_1\mapsto\chi_1];\ldots;[X_n\mapsto\chi_n]$ and all closed-respected
formulas $\psi$ with $\psi<_f\chi_1$,
\begin{align*}
\psem{\psi\sigma}\subseteq 
\sem{\psi}\widehat{\sigma}.
\end{align*}

\label{lemma:afgfpsat}
\end{lemma}

\begin{proof}
We proceed by induction over $\psi$. If $\psi=\bot$, $\psi=\top$, $\psi=p$ or $\psi=\neg p$, for $p\in P$, then
$\psi$ is closed and $\psem{\psi}=\sem{\psi}$ so that $\psem{\psi\sigma} = \psem{\psi} = \sem{\psi} =
\sem{\psi}\widehat{\sigma}$.
If $\psi = X$, then
$\psem{X\sigma} = \psem{{\sigma}(X)} =
\widehat{{\sigma}}(X) = \sem{X}_{\widehat{{\sigma}}} = \sem{X}\widehat{\sigma}$.
If $\psi = \psi_1\wedge\psi_2$, then $\psem{(\psi_1\wedge\psi_2)\sigma}
= \psem{\psi_1\sigma\wedge\psi_2\sigma} = \psem{\psi_1\sigma}\cap\psem{\psi_2\sigma}
\subseteq \sem{\psi_1}\widehat{\sigma}\cap\sem{\psi_2}\widehat{\sigma} =
\sem{\psi_1\wedge\psi_2}\widehat{\sigma}$, where the inclusion is by the induction
hypothesis. The case for disjunction is analogous.
If $\psi = \langle a\rangle\psi_1$, then 
\begin{align*}
\psem{(\langle a\rangle\psi_1)\sigma} &\subseteq
\{v\mid \exists w\in R_a(v). w\in\psem{\psi_1\sigma}\}\\
& \subseteq \{v\mid \exists w \in R_a(v). w\in\sem{\psi_1}\widehat{\sigma}\}\\
&= \sem{\langle a\rangle\psi_1}\widehat{\sigma},
\end{align*}
where the second inclusion follows from the induction hypothesis and the first
inclusion holds as follows:
Let $v\in \psem{\langle a\rangle\psi_i\sigma}$ and let $R_a(v)=\{w_1,\ldots,w_m\}$. 
There is a $(\langle a\rangle)$-rule that matches $\langle a\rangle \psi_i\sigma$ as well
as a number of $[a]$-literals from $l(v)$, 
i.e.\ that matches $\Gamma,[a]\phi_1,\ldots,[a]\phi_n,
\langle a\rangle\psi_i\sigma = l(v)$ and has the conclusion
$\{\{\phi_1,\ldots,\phi_n,\psi_i\sigma\}\}=\Sigma_j$ for some $1\leq j\leq m$.
As $w_j\in L(v)$ and $L$ is a timed-out tableau, $w_j\in\psem{\psi_i\sigma}$, as required.
If $\psi = [a]\psi_1$, then 
\begin{align*}
\psem{([a]\psi_1)\sigma} &\subseteq
\{v\mid \forall w\in R_a(v). w\in\psem{\psi_1\sigma}\}\\
& \subseteq \{v\mid \forall w \in R_a(v). w\in\sem{\psi_1}\widehat{\sigma}\}\\
&= \sem{[a]\psi_1}\widehat{\sigma},
\end{align*}
where the second inclusion follows from the induction hypothesis and the first
inclusion holds as follows:
Let $v\in \psem{[a]\psi_i\sigma}$ and let $R_a(v)=\{w_1,\ldots,w_m\}$. 
Either there is no $(\langle a\rangle)$-rule that matches $v$ in
which case $R_a(v)=\emptyset$ and we are done; or 
there is a $(\langle a\rangle)$-rule matching $[a]\psi_i\sigma$ as well
as a number of other $[a]$-literals and one $\langle a\rangle$-literal from $l(v)$, 
i.e.\ matching $\Gamma,[a]\phi_1,\ldots,[a]\phi_n,[a]\psi_i\sigma,
\langle a\rangle\phi = l(v)$ and having $\{\{\phi_1,\ldots,\phi_n,\psi_i\sigma,\phi\}\}=\Sigma_j$
as conclusion, for some $1\leq j\leq m$.
As $w_j\in L(v)$ and $L$ is a timed-out tableau, $w_j\in\psem{\psi_i\sigma}$, as required.
If $\psi = \nu Y.\,\psi_1$, then 
\begin{align*}
\psem{(\nu Y.\,\psi_1)\sigma} & =
\psem{(\psi_1[Y\mapsto \nu Y.\,\psi_1])\sigma}\\
& = \psem{\psi_1([Y\mapsto \nu Y.\,\psi_1];\sigma)}\\
& \subseteq \sem{\psi_1}\widehat{([Y\mapsto \nu Y.\,\psi_1];\sigma)},
\end{align*}
where the inclusion
is by the induction hypothesis, showing by coinduction that
$\psem{(\nu Y.\,\psi_1)\sigma} \subseteq 
\sem{\nu Y.\,\psi_1}\widehat{\sigma}$, as required.
If $\psi = \mu Y.\,\psi_1$, $\mu Y.\,\psi_1$ is closed so that
$\psem{\mu Y.\,\psi_1\sigma} = \psem{\mu Y.\,\psi_1}\subseteq\sem{\mu Y.\,\psi_1}
= \sem{\mu Y.\,\psi_1}\widehat{\sigma}$, where the inclusion is by assumption.
\end{proof}

\begin{lemma} For all closed-respected deferrals $\delta$,
all focused nodes $v\in D$, all sets of deferrals $d\subseteq d(l(v))$ and all $n\geq 0$,
\begin{center}
if $d \vdash_{l(v)} \delta$ and $v \in to(d,n)$, then $v \in \sem{\delta}$.
\end{center}
\label{lemma:aftableausat}
\end{lemma}		
Let $\delta=\alpha\sigma$ and
recall that the assumption of the lemma implies that $v \in \psem{\alpha\sigma}$.
We proceed by induction over the triple $(n,m := u_f(\alpha\sigma),\alpha)$ in
lexicographic order $<_l$. Let $[X\mapsto\mu X.\,\psi]$ and $[X_n\mapsto \theta]$ be the first and last substitutions
in $\sigma$, respectively. If $d = \emptyset$, then $\vdash_{l(v)}\alpha\sigma$
so that we cannot reach the modal cases in the upcoming case distinction -- otherwise
$\vdash_{l(v)}\langle a\rangle\alpha_1\sigma$ iff
$N(l(v),\theta)\vdash_{PL}\langle a\rangle \alpha_1\sigma$
iff $\langle a\rangle\alpha_1\sigma\in N(l(v),\theta)$, 
where $(\langle a\rangle\alpha_1,\sigma)$ is a $\theta$-deferral,
which is a contradiction since $N(l(v),\theta)$ denotes the set of formulas that are \emph{not} induced by a $\theta$-deferral. Analogously, the same holds
for $[a]\alpha_1\sigma$.
If $u_f(\alpha\sigma) = 0$, the case that $\alpha = X$ may not occur.
Recall moreover that $\delta$ is closed-respected, and hence in particular all
closed subformulas of $\alpha$ are respected.
\begin{itemize}

\item 

As $(\alpha,\sigma)$ is a deferral, $\alpha$ is open so that
$\alpha \neq \bot$, $\alpha \neq \top$, $\alpha \neq p$ and $\alpha \neq \neg p$, for $p\in P$.

\item 

If $\alpha = Y$, then let $[Y\mapsto \chi_i]$ with $\chi_i = \mu Y.\psi_i$ be the first substitution
in $\sigma$ that touches $Y$, so that
$\sigma = [X_1\mapsto\chi_1];\ldots;[Y\mapsto\chi_i];[X_{i+1}\mapsto\chi_{i+1}];\ldots;[X_n\mapsto\chi_n]$,
and $v\in \psem{\psi_i\sigma'}$, where $\sigma' = [Y\mapsto\chi_i];[X_{i+1}\mapsto\chi_{i+1}];\ldots;[X_n\mapsto\chi_n]$ and
$d \vdash_{l(v)} \psi_i\sigma'$; also $u_f(\psi_i\sigma') < m$, $(\psi_i,\sigma')$
is a deferral and $v\in to(d,n)$. By the induction hypothesis,
$v \in \sem{\psi_i\sigma'} = \sem{Y\sigma}$.

\item 

If $\alpha = \langle a\rangle\alpha_1$, then we have to show that 
there is a $w$ \sut $vR_aw$ and $w\in\sem{\alpha_1\sigma}$. 
Recall that $v\in \psem{\langle a\rangle\alpha_1\sigma}$ and let $R_a(v)=\{w_1,\ldots,w_m\}\subseteq L(v)$. 
There is a $(\langle a\rangle)$-rule that matches $\langle a\rangle \alpha_1\sigma$ as well
as a number of $[a]$-literals from $l(v)$, 
i.e.\ that matches $\Gamma,[a]\phi_1,\ldots,[a]\phi_n,
\langle a\rangle\alpha_1\sigma = l(v)$ and has the conclusion
$\{\{\phi_1,\ldots,\phi_n,\alpha_1\sigma\}\}=\Sigma_i$ for some $i$.
As $w_i\in L(v)$ and $L$ is a timed-out tableau, $w_i\in\psem{\{\phi_1,\ldots,\phi_n,\alpha_1\sigma\}}\subseteq
\psem{\alpha_1\sigma}$. We abbreviate $w_i$ by $w$ and note that we are done if $w\in\sem{\alpha_1\sigma}$.
Since $L$ is the relation of a timed-out tableau, $w\in to(d',n - 1)$ where $d' \subseteq d(w)$
and $d'\vdash_w d_{l(v)\rightsquigarrow \Gamma}$.
If $(\alpha_1,\sigma)\in d_{l(v)\rightsquigarrow \Gamma}$, we have $d'\vdash_{l(w)} \alpha_1\sigma$.
Otherwise $\vdash_{l(w)} \alpha_1\sigma$ and hence $d'\vdash_{l(w)} \alpha_1\sigma$
as well; also $(\alpha_1,\sigma)$ is a deferral. As $(n - 1, u_f(\alpha_1\sigma), \alpha_1) <_l (n,m,\alpha)$,
the induction hypothesis implies $w \in \sem{\alpha_1\sigma}$, as required.

\item 

The case for $[a]$ is analogous (cf. the proof of Lemma~\ref{lemma:afgfpsat}).

\item

If $\alpha = \alpha_1 \wedge \alpha_2$, then $v \in \psem{\alpha_1\sigma}\cap
\psem{\alpha_2\sigma}$. For any $i\in\{1,2\}$ for which $(\alpha_i,\sigma)$
is deferral, the induction hypothesis implies -- since 
$d\vdash_{l(v)}\alpha_i\sigma$, $v \in to(d,n)$,
$u_f(\alpha_i\sigma) \leq m$ and $(n, m, \alpha_i) <_l (n, m, \alpha)$ --
$v\in\sem{\alpha\sigma}$.
If $\alpha_i$ is closed, $v \in \psem{\alpha_i\sigma}=\psem{\alpha_1}$ and since
$\alpha$ and hence also $\alpha_1$ is closed-respected, $v \in \sem{\alpha_1}= \sem{\alpha_1\sigma}$.

\item

The case for $\alpha = \alpha_1 \vee \alpha_2$ is analogous to the previous case.

\item 

If $\alpha = \nu Y.\,\alpha_1$, then $\nu Y.\,\alpha_1$ is -- since fixpoint literals are alternation-free --
closed so that the induction hypothesis is not needed as we have
$v \in \psem{(\nu Y.\,\alpha_1)\sigma}=\psem{\nu Y.\,\alpha_1}$ and since
$\alpha$ is closed-respected, $v \in \sem{\nu Y.\,\alpha_1}= \sem{(\nu Y.\,\alpha_1)\sigma}$, as required.

\item

If $\alpha = \mu Y.\,\alpha_1$, then $v \in \psem{(\mu Y.\,\alpha_1)\sigma} = 
\psem{\alpha_1(\mu Y.\,\alpha_1,\sigma)}$ and $d\vdash_{l(v)} (\mu Y.\,\alpha_1)\sigma$ 
iff $d\vdash_{l(v)}\alpha_1(\mu Y.\,\alpha_1,\sigma)$. Also $(\alpha_1,(\mu Y.\,\alpha_1,\sigma))$
is deferral, $v\in to(d,n)$ and $\alpha_1$ is closed-respected so 
that the induction hypothesis implies $v \in \sem{\alpha_1(\mu Y.\,\alpha_1,\sigma)} =
\sem{(\mu Y.\,\alpha_1)\sigma}$, as required.

\end{itemize}
This finishes the proof.
\qed

\begin{lemma}
All closed fixpoint literals are respected.
\label{lemm:affpresp}
\end{lemma}
\begin{proof}
Let $\eta X.\,\psi$ be a closed fixpoint literal. 
We proceed by induction over the depth of nesting of closed fixpoint literals $n=\mathit{cfd}(\eta X.\,\psi)$
in $\eta X.\,\psi$. If $n=1$, then
$\psi$ contains no closed fixpoint literals and is hence closed-respected so that 
if $\eta=\mu$, case a) and if $\eta=\nu$, case b) of Lemma~\ref{lemm:affps} applies 
and finishes the case. If $n>1$,
then any closed fixpoint literal $\eta Y.\phi\leq \psi$ has a depth of nesting of closed fixpoint literals 
less than $n$ and is respected by induction. Thus
$\psi$ is closed-respected so that Lemma~\ref{lemm:affps}
finishes the proof.
\end{proof}

\Ni\emph{Proof of Lemma~\ref{lemma:truth}:}
We proceed by induction over $\psi$. If $\psi=\bot$, $\psi=\top$, $\psi=p$ or $\psi=\neg p$, for
$p\in P$, then $\psem{\psi}=\sem{\psi}$ by definition. For the propositional connectives, the
inductive step is straightforward. If $\psi=\langle a\rangle \psi_1$, then note there
is for any state $v\in\psem{\langle a\rangle \psi_1}$, and any focused node $(\Delta,d)$,
a rule
\begin{align*}
(\Gamma,[a]\phi_1,\ldots,[a]\phi_n,\langle a\rangle\psi_1\,/\,\phi_1,\ldots,\phi_n,\psi_1)
\end{align*}
matching $\Delta$, i.e.\ with $\Delta=\Gamma,[a]\phi_1,\ldots,[a]\phi_n,\langle a\rangle\psi_1$. As
we operate in a Kripke structure over a timed-out tableau, $(\Delta,d)\in \mathit{to}(d(\Delta),m)$ so that
there is a focused node $(\Theta,d_{\Delta\rightsquigarrow \Theta})\in L(\Delta,d)$ with $(\Theta,d_{\Delta\rightsquigarrow \Theta})\in R_a(\Delta,d)$ and
$(\Theta,d_{\Delta\rightsquigarrow \Theta})\in \psem{\{\phi_1,\ldots,\phi_n,\psi_1\}}\subseteq\psem{\psi_1}$.
The induction hypothesis finishes the case. The case where $\psi=[a]\psi_1$ is analogous.
If $\psi=\eta X.\,\psi_1$, then Lemma~\ref{lemm:affpresp} finishes the case.
\qed\\

\Ni\emph{Proof of Theorem~\ref{prop:exptime}:}
  The algorithm terminates and as we have seen, it is sound
  and complete, thus it decides the problem.  Let $n=|\phi_0|$ where
  $\phi_0$ denotes the input formula. The algorithm consists
  of a loop which is repeated at most $a:=2^n$ times since any of the
  at most $a$ nodes from $\nodes$ has been expanded after at most $a$
  expansion steps.  The body of the loop consists of one expansion
  step and one optional propagation step. Since we are interested in worst-case
	performance of the algorithm, we ignore the optional propagation step. Since 
	modal, propositional and fixpoint literal expansion is
  implementable in \ExpTime, the expansion step runs in \ExpTime.  We
  convince ourselves that the propagation step runs in \ExpTime as
  well, which intuitively follows from the fact that
	propagation computes fixpoints over $G$
	where $|G|\leq 3^n$.
	We consider the computation of the set $E_G$ and
  note that analysis of the computation of $A_G$ is analogous. As
  $E_G=\nu (X\mapsto \mu(\hat{f}_X))=(X\mapsto \mu(\hat{f}_X))^m(\mathbf{C}_G)$ for some $m\leq 3^n$, the
  computation consists of at most $3^n$ computations of
  $\mu(\hat{f}_{X})$, each for some $X\subseteq \mathbf{C}_G$. A single
  computation of $\mu(\hat{f}_X)=(\hat{f}_X)^o(\emptyset)$ for some $o\leq 3^n$ consists of 
	at most $3^n$ computations of $\hat{f}_X(Y)$, each for some $Y\subseteq \mathbf{C}_G$.
	The computation of $\hat{f}_{X}(Y)$ checks for each $(\Gamma,d)\in \mathbf{C}_G$
	whether there	is a conclusion $(\Theta,d_{\Gamma\rightsquigarrow \Theta})\in X\cap Y$ (or $(\Theta,d(\Theta))\in X$, if $d=\emptyset$) for each rule that matches $\Gamma$.
	Propagation thus runs in time at most $(3^n)^c=3^{c\cdot n}$ for some constant $c$, and, therefore, the algorithm runs in \ExpTime;
	modal expansion can be implemented in time $2^{\lO(n)}$ in the relational case
  so that the runtime of the algorithm is bounded by $2^{\lO(n)}$. 
	\qed

\subsection{Details on New Benchmark Formulas in Section~\ref{section:cool}}
We define the formulas $c(x,n)$ by putting $c(x,n):=c_n(x,n)$, where $c_n(x,i)$ is defined recursively as
\begin{align*}
c_n(x,i) &= (\neg x_{n-i} \wedge AX \; x_{n-i} \wedge \psi_n(x,i-1)) \vee (x_{n-i} \wedge AX\; \neg x_{n-i} \wedge c_n(x,i-1))\\
\psi_n(x,i) &= (\neg x_{n-i} \vee AX \; x_{n-i}) \wedge (x_{n-i} \vee AX\; \neg x_{n-i}) \wedge \psi_n(x,i-1).
\end{align*}

\end{document}